\documentclass[journal, headsepline]{IEEEtran}

\usepackage{amsmath,amsfonts,amssymb,amsthm}
\usepackage{mathrsfs}
\usepackage{comment,blkarray}
\usepackage{multirow,bigdelim}
\usepackage{cite}
\usepackage[ruled,commentsnumbered, vlined]{algorithm2e}
\usepackage{tikz}
\usetikzlibrary{arrows,automata}
\usepackage[latin1]{inputenc}
\usepackage{verbatim}
\usepackage{graphicx}
\usepackage{cases}
\usepackage{booktabs}
\usepackage{caption}												
\usepackage{subcaption}	
\usepackage{etoolbox}

\usepackage{enumitem}
\usepackage{mathtools}

\makeatletter

\newtheorem{thm}{Theorem}
\newtheorem{prop}{Proposition}
\theoremstyle{definition}
\newtheorem{example}{Example}
\allowdisplaybreaks

\begin{document}

\title{Moving Target SAR Imaging Using Planar Arrays And Multidimensional Chinese Remainder Theorem (MD-CRT)--Part II: Two Subarray Designs}

\author{Guangpu Guo, \IEEEmembership{Graduate Student Member}, \IEEEmembership{IEEE}, and Xiang-Gen Xia, \IEEEmembership{Fellow}, \IEEEmembership{IEEE} 
       
\thanks{G. Guo and X.-G. Xia are with the Department of Electrical and Computer Engineering,
  University of Delaware, Newark, DE 19716, USA
  (e-mails: guangpu@udel.edu and xxia@ee.udel.edu).
  This work was supported in part by the National Science Foundation (NSF) under Grant CCF-2246917.
}
}

\maketitle

\begin{abstract}
Based on the framework proposed in Part I, the Part II of this two-part paper investigates two-subarray designs for moving target SAR imaging using planar antenna arrays and the multidimensional Chinese remainder theorem (MD-CRT). In this Part II, we focus on the performance analysis and the detailed two planar subarray designs. In particular, we study a common-scaling two-subarray design, under which the two subarrays share the same scaling factor in the MD-CRT formulation. Under this design, ambiguity resolution can be performed on a common integer frequency vector. As a result, the same unambiguous range as in the general two-subarray framework in Part~I is preserved, while the sufficient conditions for robust recovery become weaker and the corresponding reconstruction error bounds become tighter.
Within this common-scaling design, we compare the proposed planar array framework with a conventional separated scheme, in which the motion-induced cross-range shift is recovered by a one-dimensional CRT-based method and the target height is estimated by cross-track interferometric processing. Under the same platform size and minimum antenna spacing constraints, the proposed planar array framework can realize the common-scaling design, whereas the corresponding one-dimensional non-uniform linear array scheme does not admit such a design. With this design, the planar array framework leads to a weaker sufficient condition for robust recovery and thus performs better in moving target imaging.
We also compare several planar array designs under fixed platform size and minimum antenna spacing. The analysis shows that recovery performance depends not only on the number of antennas but also on the array geometry. In particular, non-separable planar array geometries can provide better robustness than separable ones when their antenna numbers are comparable. Numerical results validate the theoretical analysis.
\end{abstract}


\section{Introduction}\label{s1}

In Part~I of this two-part paper, we developed a general framework for moving target SAR imaging with planar antenna arrays~\cite{part1}. 
For a target moving over a three-dimensional terrain (or in air), the motion-induced cross-range shift and the target height were formulated in a unified two-dimensional vector form and were shown to admit a matrix-modulus vector remainder representation after two-dimensional discrete Fourier transform (2D-DFT) processing across the planar array~\cite{ar}. 
This leads to a general multi-subarray ambiguity resolution framework for enlarging the unambiguous range, based on the robust multidimensional Chinese remainder theorem (MD-CRT)~\cite{MD1,MD2,gMDCRT,mstage-mdcrt}. 
Part~I also established sufficient conditions and explicit reconstruction error bounds for robust recovery under quantization and noise.

While Part~I established the general framework, it did not address in detail how planar subarrays should be designed for good performance. 
Compared with the one-dimensional non-uniform linear arrays studied in~\cite{gangli}, the two-dimensional planar arrays considered in this two-part paper involve matrix-valued moduli rather than scalar moduli, and therefore admit substantially more design degrees of freedom. 
This additional freedom also makes the design problem considerably more challenging. 
Part~I showed that, in the two-dimensional setting, the choice of the associated 2D-DFT matrix modulus affects the propagation of vector remainder errors, and an optimal construction was derived under fixed array geometry and antenna number constraints. 
However, even under this optimal modulus construction, the recovery performance still depends on the array geometry itself, namely, on the lattice generator matrix that defines the planar array. 
The focus of this Part~II is therefore to investigate how planar subarrays should be designed, and how different array geometries behave under the same physical constraints.

In this Part~II, we focus on the two-subarray case, which is the simplest nontrivial setting that already captures the main design issues of interest while allowing the analysis to be carried out analytically. 
Within this setting, we first study a common-scaling design, under which the two subarrays share the same scaling factor in the corresponding MD-CRT formulation. 
Under this design, ambiguity resolution can be carried out on a common integer frequency vector. 
As a result, the recovery procedure becomes simpler than that in the general formulation in Part~I, since no additional scaling is needed to convert the congruence system into an integer-valued form, while the same unambiguous range as in the general two-subarray framework in Part~I is preserved. 
Furthermore, in both noiseless and noisy settings, the sufficient conditions for robust recovery become weaker, and the resulting reconstruction error bounds become smaller than those in the general formulation.
Therefore, the common-scaling design provides better robustness. 
Moreover, this common-scaling design can be realized naturally in the planar array setting considered here, whereas for the one-dimensional non-uniform linear array scheme in~\cite{gangli}, it is not possible when the two subarrays have the same number of antenna elements.

Based on this common-scaling design, we first compare the proposed planar array framework with a conventional separated scheme. 
In the separated scheme, the motion-induced cross-range shift is recovered by a one-dimensional CRT-based method with two non-uniform linear arrays~\cite{gangli}, while the target height is estimated by cross-track interferometric processing~\cite{xiaoweili}, which can be thought of as a two-point DFT in the vertical direction. 
Under the same platform size and minimum antenna spacing constraints, the two schemes have the same unambiguous range. 
However, because the proposed planar array framework can realize the common-scaling design whereas the corresponding one-dimensional non-uniform linear array scheme cannot, the planar array framework offers clear robustness advantages over the separated scheme. 
In particular, in the noiseless setting, the proposed planar array framework avoids the additional condition required in the one-dimensional separated scheme for robust recovery and yields a smaller reconstruction error bound for the cross-range shift.
In the noisy setting, numerical results further confirm the robustness advantage of the proposed framework for recovering both the cross-range shift and the height.

We then compare several planar array designs under fixed platform size and minimum antenna spacing. 
In many existing planar array designs, the geometry is chosen in a separable form, for example, a rectangular structure with a diagonal lattice generator matrix~\cite{FDA_MIMO_Planar,Incremental3DSAR}. 
Such designs are simple and convenient, since the two spatial directions can be treated separately. 
However, a two-dimensional matrix-valued design has much more choices, which raises the question of whether a non-separable planar array design can provide better performance under the same physical constraints. 
The analysis and simulations show that recovery performance depends not only on the number of antennas, but also on the array geometry. 
In the noiseless case, under the considered constraints, non-separable planar array designs can achieve a larger unambiguous range than separable ones. 
In the noisy case, non-separable common-scaling planar array designs can provide better robustness than separable common-scaling planar array designs when the latter use more antennas and occupy a larger platform area. 
This advantage comes from the greater design flexibility of non-separable array designs under the minimum antenna spacing constraint. 
Further theoretical studies on the difference between non-separable and separable matrix moduli in robust MD-CRT can be found in~\cite{mstage-mdcrt}.

The remainder of this paper is organized as follows. 
Section~\ref{s2} briefly reviews the notation and problem setting needed from Part~I. 
Section~\ref{s3} introduces the common-scaling two-subarray design and studies its recovery properties in both noiseless and noisy settings. 
Section~\ref{s4} compares the proposed planar array framework with the separated processing scheme and then compares several planar array designs under the same physical constraints. 
Section~\ref{s5} presents numerical results, and Section~\ref{s6} concludes the paper.

\section{Problem Setting and General Two-Subarray Formulation}\label{s2}

In this section, we briefly review the general framework in Part I when only two subarrays are considered, which will be used throughout this Part II. 
The single-array case is included as the special case in which only one matrix modulus is used. All the notations in this Part II follow from Part I, unless otherwise specified.

\subsection{Problem Formulation for Two Subarrays and Unambiguous Range}

Consider a moving target observed by two planar subarrays in a SAR system. 
The task is to jointly recover the motion-induced cross-range shift and the target height. 

After the standard SAR processing and the compensation of the image location term, the antenna-dependent image obtained from the $\mathbf{u}$-th antenna of the $j$-th planar subarray, $j=1,2$, can be written as
\begin{equation}\label{eq:s2_signal_model}
S_{\mathbf u,j}(n)
=
C_0
\exp\!\left(
j\frac{2\pi}{R\lambda}
(\mathbf M_j^\top\mathbf g)^\top\mathbf u
\right)
\delta\!\big(n-n_{x_0}-\Delta_{\rm shift}\big),
\end{equation}
where $C_0$ is the antenna-independent phase term as shown in the equation (12) in Part I \cite{part1}, $R$ is the reference range, $\lambda$ is the carrier wavelength, and $\mathbf u\in\mathbb Z^2$ is the index of an antenna element in the $j$-th planar subarray. 
The matrix $\mathbf M_j\in\mathbb R^{2\times2}$ is the lattice generator matrix of the $j$-th planar subarray, so that the corresponding antenna location of the $\mathbf{u}$-th antenna on the platform is given by
\[
\begin{bmatrix}
x_{\mathbf u,j}\\
z_{\mathbf u,j}
\end{bmatrix}
=
\mathbf M_j\mathbf u.
\]
The unknown parameter vector is
\begin{equation}\label{eq:s2_g_def}
\mathbf g
\triangleq
\begin{bmatrix}
-\Delta_{\rm shift}\Delta_x\\
z_0
\end{bmatrix},
\end{equation}
where $\Delta_{\rm shift}$ is the motion-induced cross-range shift, $\Delta_x$ is the cross-range resolution, and $z_0$ is the target height. 
Therefore, recovering $\mathbf g$ is equivalent to jointly recovering the shift and the height.

For the $j$-th planar subarray, $j=1,2$, Part~I showed that, under a fixed array geometry $\mathbf{M}_j$ and antenna number constraint, a natural and optimal choice of the associated 2D-DFT matrix modulus is
\begin{equation}\label{eq:s2_Nj_def}
\mathbf N_j
=
\alpha_j\,\operatorname{adj}(\mathbf M_j^\top),
\qquad j=1,2,
\end{equation}
where $\alpha_j$ is chosen so that $\mathbf N_j\in\mathbb Z^{2\times2}$, i.e., $\mathbf{N}_j$ is an integer matrix. 
Under this choice, the corresponding frequency vector in \eqref{eq:s2_signal_model} is
\begin{equation}\label{eq:s2_fj_def}
\mathbf f_j
=
\frac{\alpha_j\det(\mathbf M_j)}{R\lambda}\,\mathbf g,
\qquad j=1,2.
\end{equation}
Equivalently, defining the scaling factor
\begin{equation}\label{eq:s2_betaj_def}
\beta_j
\triangleq
\frac{R\lambda}{\alpha_j\det(\mathbf M_j)},
\qquad j=1,2,
\end{equation}
we may write
\begin{equation}\label{eq:s2_g_fj_relation}
\mathbf g
=
\beta_j\,\mathbf f_j,
\qquad j=1,2.
\end{equation}

For the $j$-th planar subarray, we perform the 2D-DFT of $S_{\mathbf u,j}(n)$ in \eqref{eq:s2_signal_model} with respect to the antenna index $\mathbf u$ over one fundamental parallelepiped (FPD) of $\mathbf N_j^\top$, i.e., over the index set $\mathcal N(\mathbf N_j^\top)$. 
Here, for a nonsingular integer matrix $\mathbf A\in\mathbb Z^{2\times 2}$, the associated FPD is defined by
\[
\mathcal N(\mathbf A)
=
\left\{
\mathbf k\in\mathbb Z^2:\,
\mathbf k=\mathbf A\mathbf x,\;
\mathbf x\in[0,1)^2
\right\}.
\]
The corresponding 2D-DFT frequency vector index then lies in $\mathcal N(\mathbf N_j)$.
If $\mathbf f_j$ is integer-valued, then the 2D-DFT peak yields the vector remainder $\mathbf r_j$ of $\mathbf f_j$ modulo $\mathbf N_j$, i.e.,
\begin{equation}\label{eq:s2_exact_congruence}
\mathbf f_j
\equiv
\mathbf r_j
\mod \mathbf N_j,
\qquad
\mathbf r_j\in\mathcal N(\mathbf N_j),
\qquad j=1,2.
\end{equation}
Using \eqref{eq:s2_g_fj_relation}, this is equivalently written as
\begin{equation}\label{eq:s2_g_congruence}
\mathbf g
\equiv
\beta_j\,\mathbf r_j
\mod \beta_j\mathbf N_j,
\qquad j=1,2.
\end{equation}
Thus, the two subarrays produce two folded observations of the same vector $\mathbf g$, and these two observations can be combined through the MD-CRT~\cite{MD1,gMDCRT} to reconstruct $\mathbf g$.

Assume that \eqref{eq:s2_g_congruence} is in the integer-valued form required by the MD-CRT, and let
\begin{equation}\label{eq:s2_R_def}
\mathbf R
\triangleq
\operatorname{lcrm}(\beta_1\mathbf N_1,\beta_2\mathbf N_2),
\end{equation}
where $\operatorname{lcrm}(\cdot,\cdot)$ denotes a least common right multiple (lcrm) of the two scaled matrix moduli. 
Then, by the MD-CRT~\cite{MD1}, $\mathbf g$ can be uniquely recovered from \eqref{eq:s2_g_congruence} if and only if $\mathbf g\in\mathcal N(\mathbf R)$. 
Hence, $\mathcal N(\mathbf R)$ defines an \textit{unambiguous range} of the two-subarray system in the sense that any integer vector $\mathbf g\in\mathcal N(\mathbf R)$ can be uniquely recovered from \eqref{eq:s2_g_congruence}. 
Although the shape of this range depends on the chosen lcrm, its size, i.e., the number of integer vectors in $\mathcal N(\mathbf R)$, is $|\det(\mathbf R)|=|\mathcal N(\mathbf R)|$ that does not depend on the choice of an lcrm $\mathbf{R}$ and is uniquely determined by the two matrix moduli $\beta_1\mathbf N_1$ and $\beta_2\mathbf N_2$.

\subsection{Robust Recovery Under Quantization and Noise}

We first consider the quantization of the unknown frequency vector $\mathbf f_j$.
Throughout this paper, quantization refers to frequency grid quantization caused by evaluating the DFT on a discrete frequency grid.
In general, $\mathbf f_j$ is not necessarily integer-valued. 
Following Part~I, we define
\begin{equation}\label{eq:s2_fj_int}
\mathbf f_j^{\mathrm{int}}
\triangleq
[\mathbf f_j],
\qquad
\mathbf e_j
\triangleq
\mathbf f_j^{\mathrm{int}}-\mathbf f_j,
\qquad j=1,2,
\end{equation}
where $[\cdot]$ denotes componentwise nearest-integer rounding defined in (54) in Part~I. 
Then,
\begin{equation}\label{eq:s2_ej_bound}
\|\mathbf e_j\|_\infty \le \frac{1}{2},
\qquad j=1,2.
\end{equation}
Accordingly, the detected 2D-DFT peak $\mathbf r_j\in\mathcal N(\mathbf N_j)$ is 
\begin{equation}\label{eq:s2_quantized_congruence}
\mathbf f_j^{\mathrm{int}}
\equiv
\mathbf r_j
\mod \mathbf N_j,
\qquad j=1,2.
\end{equation}
Using $\mathbf g=\beta_j\mathbf f_j$ and the approximation model $\mathbf f_j^{\mathrm{int}}=\mathbf f_j+\mathbf e_j$, we may write
\begin{equation}\label{eq:s2_g_congruence_quantized}
\mathbf g
\equiv
\beta_j\mathbf r_j-\beta_j\mathbf e_j
\mod \beta_j\mathbf N_j,
\end{equation}
for $j=1,2$.
Thus, the recovery problem can be viewed as a matrix congruence system for $\mathbf g$ with erroneously detected vector remainders $\beta_j\mathbf r_j$. 

To apply robust MD-CRT, we first convert the recovery problem in \eqref{eq:s2_g_congruence_quantized} into an integer-valued congruence form. 
Let $k$ be a positive scalar such that, for $j=1,2$, the entries of
\begin{equation}\label{eq:s2_Abj_def}
\mathbf A_j
\triangleq
k\beta_j\mathbf N_j
\in\mathbb Z^{2\times2},
\qquad
\mathbf b_j
\triangleq
k\beta_j\mathbf r_j
\in\mathbb Z^2,
\end{equation}
are integers.
We also define
\begin{equation}\label{eq:s2_x_def}
\mathbf x
\triangleq
[k\mathbf g]\in\mathbb Z^2,
\qquad
\boldsymbol\delta
\triangleq
\mathbf x-k\mathbf g,
\end{equation}
and then
\begin{equation}\label{eq:s2_delta_bound}
\|\boldsymbol\delta\|_\infty \le \frac{1}{2}.
\end{equation}
Thus, using \eqref{eq:s2_g_congruence_quantized}, we obtain
\begin{equation}\label{eq:s2_x_congruence}
\mathbf x
=
\mathbf A_j\mathbf n_j+\mathbf b_j+\boldsymbol\eta_j,
\qquad j=1,2,
\end{equation}
where
\begin{equation}\label{eq:s2_eta_def}
\boldsymbol\eta_j
\triangleq
\boldsymbol\delta-k\beta_j\mathbf e_j.
\end{equation}
Hence,
\begin{equation}\label{eq:s2_eta_bound}
\|\boldsymbol\eta_j\|_\infty
\le
\frac{1}{2}
+
\frac{k|\beta_j|}{2},
\qquad j=1,2.
\end{equation}

Let $\mathbf G$ be a greatest common left divisor (gcld) of $\mathbf A_1$ and $\mathbf A_2$, and define
\begin{equation}\label{eq:s2_Lambda_def}
\Lambda
\triangleq
\lambda_{\mathcal L(\mathbf G)},
\end{equation}
where $\lambda_{\mathcal L(\mathbf G)}$ is the minimum nonzero $\ell_\infty$-distance of the lattice generated by $\mathbf G$. 
Also define
\begin{equation}\label{eq:s2_beta_max}
\beta_{\max}
\triangleq
\max\{|\beta_1|,|\beta_2|\}.
\end{equation}

As in Part~I, robust recovery does not change the size of the unambiguous range determined in the integer-valued case. 
It only imposes additional conditions under which robust reconstruction is guaranteed. 
Following Theorem 2 in Part I \cite{part1}, the corresponding noiseless robust recovery result can be stated as follows.

If
\begin{equation}\label{eq:s2_noiseless_condition}
\frac{1}{2}
+
\frac{k\beta_{\max}}{2}
<
\frac{\Lambda}{4},
\end{equation}
then $\mathbf g$ can be robustly recovered from the two detected vector remainders. 
Moreover, the reconstruction error satisfies
\begin{equation}\label{eq:s2_noiseless_bound}
\|\hat{\mathbf g}-\mathbf g\|_\infty
\le
\frac{\beta_{\max}}{2}
+
\frac{1}{k}.
\end{equation}

We next consider the noisy case. 
Suppose that due to the noise, the detected vector remainder $\tilde{\mathbf r}_j$ for the $j$-th subarray has error,
\begin{equation}\label{eq:s2_rtilde_def}
\tilde{\mathbf r}_j
=
\mathbf r_j+\Delta\mathbf r_j,
\qquad j=1,2,
\end{equation}
where $\Delta\mathbf r_j$ denotes the vector remainder error caused by noise. 
Define
\begin{equation}\label{eq:s2_btilde_def}
\tilde{\mathbf b}_j
\triangleq
k\beta_j\tilde{\mathbf r}_j
=
\mathbf b_j+k\beta_j\Delta\mathbf r_j,
\qquad j=1,2.
\end{equation}
Then the corresponding integer-valued congruence system becomes
\begin{equation}\label{eq:s2_x_congruence_noise}
\mathbf x
=
\mathbf A_j\mathbf n_j+\tilde{\mathbf b}_j+\tilde{\boldsymbol\eta}_j,
\qquad j=1,2,
\end{equation}
where
\begin{equation}\label{eq:s2_eta_noise_def}
\tilde{\boldsymbol\eta}_j
=
\boldsymbol\delta-k\beta_j\mathbf e_j-k\beta_j\Delta\mathbf r_j.
\end{equation}
Hence,
\begin{equation}\label{eq:s2_eta_noise_bound}
\|\tilde{\boldsymbol\eta}_j\|_\infty
\le
\frac{1}{2}
+
\frac{k|\beta_j|}{2}
+
k|\beta_j|\,\|\Delta\mathbf r_j\|_\infty,
\qquad j=1,2.
\end{equation}

Following Theorem 3 in Part I \cite{part1}, the corresponding noisy robust recovery result can be stated as follows.

If
\begin{equation}\label{eq:s2_noisy_condition}
\frac{1}{2}
+
\frac{k\beta_{\max}}{2}
+
k\max_{j=1,2}\Big(|\beta_j|\,\|\Delta\mathbf r_j\|_\infty\Big)
<
\frac{\Lambda}{4},
\end{equation}
then $\mathbf g$ can be robustly recovered under noisy vector remainder detection. 
Moreover, the reconstruction error satisfies
\begin{equation}\label{eq:s2_noisy_bound}
\|\hat{\mathbf g}-\mathbf g\|_\infty
\le
\frac{\beta_{\max}}{2}
+
\max_{j=1,2}\Big(|\beta_j|\,\|\Delta\mathbf r_j\|_\infty\Big)
+
\frac{1}{k}.
\end{equation}

\section{A Common-Scaling Two-Subarray Design}\label{s3}

In this section, we consider a special two-subarray design within the general two-subarray framework reviewed in Section~\ref{s2}. 
Its key feature is that the two planar subarrays share the same scaling factor $\beta_j$ in the corresponding MD-CRT formulation. 
Under this condition, ambiguity resolution can be performed on a common integer-valued frequency vector $\mathbf f$. 
As a result, the same unambiguous range as in the general two-subarray framework is preserved, while the sufficient conditions for robust recovery become weaker and the resulting reconstruction error bounds become tighter, leading to better robustness. 
This common-scaling design provides the basis for the theoretical comparisons in the remainder of this paper. 
We also note that the common-scaling design is not possible for the one-dimensional non-uniform linear array scheme in~\cite{gangli} when the two subarrays have the same number of antenna elements. This point will be detailed in Section~\ref{s4.1}.

\subsection{Common-Scaling Formulation and Noiseless Recovery}

We impose the common-scaling condition
\begin{equation}\label{eq:common_beta}
\beta_1=\beta_2=\beta.
\end{equation}
By the definition of $\beta_j$ in \eqref{eq:s2_betaj_def},
\[
\beta_j=\frac{R\lambda}{\alpha_j\det(\mathbf M_j)},
\]
this condition in \eqref{eq:common_beta} is satisfied if the two planar subarray generator matrices $\mathbf M_1$ and $\mathbf M_2$ are chosen such that
\[
\alpha_1\det(\mathbf M_1)=\alpha_2\det(\mathbf M_2),
\]
while the corresponding integer matrix moduli $\mathbf N_1$ and $\mathbf N_2$ in \eqref{eq:s2_Nj_def} are distinct.
Note that the detailed comparison with the 1D linear non-uniform subarray case proposed in \cite{gangli} will be given in next section.
Under the condition in \eqref{eq:common_beta}, the two subarrays share the same frequency vector, i.e.,
\begin{equation}\label{eq:common_f}
\mathbf f_1=\mathbf f_2=\mathbf f=\frac{1}{\beta}\mathbf g.
\end{equation}

We next show how the common-scaling design reduces the recovery problem to that of a single integer frequency vector. 
Since the 2D-DFT is evaluated on an integer frequency vector grid, we model the detected vector remainders by those of an integer approximation of $\mathbf f$, rather than by $\mathbf f$ itself.
We therefore define
\begin{equation}\label{eq:fint_def}
\mathbf f^{\mathrm{int}} \triangleq [\mathbf f]\in\mathbb Z^2,
\end{equation}
and introduce the associated quantization error
\begin{equation}\label{eq:e_def_special}
\mathbf e \triangleq \mathbf f^{\mathrm{int}}-\mathbf f,
\end{equation}
and thus,
\begin{equation}\label{eq:e_bound_special}
\|\mathbf e\|_\infty \le \frac{1}{2}.
\end{equation}
For each subarray $j$, $j=1,2$, the detected 2D-DFT peak yields the vector remainder $\mathbf r_j$ of the same integer vector $\mathbf f^{\mathrm{int}}$ modulo $\mathbf N_j$, i.e.,
\begin{equation}\label{eq:fint_congruence_special}
\mathbf f^{\mathrm{int}} \equiv \mathbf r_j \mod \mathbf N_j,
\qquad j=1,2,
\end{equation}
where $\mathbf r_j\in\mathcal N(\mathbf N_j)$.

Therefore, under the common-scaling condition in \eqref{eq:common_beta}, the two congruences in \eqref{eq:fint_congruence_special} involve the same integer vector $\mathbf f^{\mathrm{int}}$. 
We may thus first recover $\mathbf f^{\mathrm{int}}$ from these two congruences and then recover $\mathbf g$ through a simple scaling by $\beta$ as follows.

Under the common-scaling condition, the congruence system in \eqref{eq:fint_congruence_special} is an exact integer-valued congruence system for the common integer vector $\mathbf f^{\mathrm{int}}$, and the corresponding vector remainders are error-free. 

By the MD-CRT~\cite{MD1}, see also Proposition~1 in Part~I, $\mathbf f^{\mathrm{int}}$ can be uniquely recovered from \eqref{eq:fint_congruence_special} if and only if
\[
\mathbf f^{\mathrm{int}}\in\mathcal N(\mathbf R_f).
\]
Hence, $\mathcal N(\mathbf R_f)$ defines an unambiguous range for recovering $\mathbf f^{\mathrm{int}}$ in the common-scaling design.

Once $\mathbf f^{\mathrm{int}}$ is reconstructed from \eqref{eq:fint_congruence_special}, we estimate $\mathbf g$ by
\begin{equation}\label{eq:g_hat_from_fint}
\hat{\mathbf g}
\triangleq
\beta\,\mathbf f^{\mathrm{int}}.
\end{equation}
Since $\mathbf g=\beta\mathbf f$ and $\mathbf f^{\mathrm{int}}=\mathbf f+\mathbf e$, it follows that
\begin{equation}\label{eq:g_hat_error_special}
\hat{\mathbf g}-\mathbf g
=
\beta\,\mathbf e.
\end{equation}
Therefore, in the noiseless case, the reconstruction error satisfies
\begin{equation}\label{eq:g_error_bound_special}
\|\hat{\mathbf g}-\mathbf g\|_\infty
\le
\frac{|\beta|}{2}.
\end{equation}

The above discussion leads to the following noiseless recovery result.

\begin{thm}[Common-Scaling Noiseless Recovery]\label{thm:common_scaling_noiseless}
Assume that the common-scaling condition in \eqref{eq:common_beta} holds, and let
$\mathbf R_f=\operatorname{lcrm}(\mathbf N_1,\mathbf N_2)$.
If
$\mathbf f^{\mathrm{int}}\in\mathcal N(\mathbf R_f)$,
then $\mathbf f^{\mathrm{int}}$ can be uniquely recovered from the two congruences in \eqref{eq:fint_congruence_special}. Moreover, the estimate
$\hat{\mathbf g}=\beta\,\mathbf f^{\mathrm{int}}$
satisfies the reconstruction error bound in \eqref{eq:g_error_bound_special}.
\end{thm}

We now compare the common-scaling design with the general two-subarray framework in Section~\ref{s2}. 
In the general formulation, the unambiguous range for $\mathbf g$ is determined by the FPD of
\[
\operatorname{lcrm}(\beta_1\mathbf N_1,\beta_2\mathbf N_2).
\]
Under the common-scaling condition $\beta_1=\beta_2=\beta$, this becomes
\begin{equation}\label{eq:lcrm_common_scaling}
\operatorname{lcrm}(\beta\mathbf N_1,\beta\mathbf N_2)
=
\beta\,\operatorname{lcrm}(\mathbf N_1,\mathbf N_2)
=
\beta\,\mathbf R_f.
\end{equation}
Hence, the unambiguous range for $\mathbf g$ is exactly the same as that in the general two-subarray framework.

The advantage of the common-scaling design is that the recovery becomes simpler. 
Instead of first converting the congruence system for $\mathbf g$ in \eqref{eq:s2_g_congruence_quantized} into the integer-valued form \eqref{eq:s2_Abj_def}--\eqref{eq:s2_x_def} and then recovering $[k\mathbf g]$ by robust MD-CRT, we directly recover the common integer vector $\mathbf f^{\mathrm{int}}$ from \eqref{eq:fint_congruence_special}, and then obtain $\mathbf g$ by the scaling $\hat{\mathbf g}=\beta\,\mathbf f^{\mathrm{int}}$ in \eqref{eq:g_hat_from_fint}. 
As a result, the noiseless error bound in \eqref{eq:g_error_bound_special} no longer contains the extra $\frac{1}{k}$ term appearing in \eqref{eq:s2_noiseless_bound}, and the sufficient condition \eqref{eq:s2_noiseless_condition} is no longer needed. 
The reason is that, in the common-scaling noiseless setting, recovery is based on the exact congruence system \eqref{eq:fint_congruence_special} for the common integer vector $\mathbf f^{\mathrm{int}}$, to which the classical MD-CRT applies directly, rather than applying robust MD-CRT to an erroneous congruence system for $[k\mathbf g]$. 
This is important in practice under platform size and minimum antenna spacing constraints, where satisfying \eqref{eq:s2_noiseless_condition} may require a relatively large number of antennas. 
Therefore, the common-scaling design preserves the same unambiguous range as the general framework while providing a simpler recovery procedure, weaker sufficient conditions for robust recovery, and tighter reconstruction error bounds, which together lead to better robustness.

\subsection{Noisy Recovery Under the Common-Scaling Design}

We next extend the common-scaling analysis to the noisy case, where the vector remainders detected from the two subarrays are perturbed by noise.
Specifically, let the detected vector remainders be
\begin{equation}\label{eq:r_tilde_special}
\tilde{\mathbf r}_j
=
\mathbf r_j+\Delta\mathbf r_j,
\qquad j=1,2,
\end{equation}
where $\Delta\mathbf r_j$ denotes the error vector in the $j$-th detected vector remainder caused by noise.
The true integer vector $\mathbf f^{\mathrm{int}}$ still satisfies the congruence system in \eqref{eq:fint_congruence_special}, but only the noisy observations $\tilde{\mathbf r}_1$ and $\tilde{\mathbf r}_2$ are available.
The task in the noisy case is therefore to reconstruct $\mathbf f^{\mathrm{int}}$ robustly from these erroneous vector remainders.

The noisy recovery problem can still be analyzed within the robust MD-CRT framework reviewed in Section~\ref{s2}. 
In the common-scaling design, robust reconstruction is performed directly for the integer vector $\mathbf f^{\mathrm{int}}$ from its two erroneous vector remainders modulo $\mathbf N_1$ and $\mathbf N_2$. 
In robust MD-CRT, the unambiguous range for $\mathbf f^{\mathrm{int}}$ can still be $\mathcal N(\mathbf R_f)$ when $\mathbf N_1$ and $\mathbf N_2$ are designed appropriately, as discussed in~\cite{mstage-mdcrt}, where
\[
\mathbf R_f=\operatorname{lcrm}(\mathbf N_1,\mathbf N_2).
\]
Accordingly, the unambiguous range for $\mathbf g$ remains the same as that in the noiseless common-scaling case.

Let $\mathbf G_f$ be a gcld of $\mathbf N_1$ and $\mathbf N_2$, and similar to \eqref{eq:s2_Lambda_def}, define
\begin{equation}\label{eq:Lambda_f_def}
\Lambda_f
\triangleq
\lambda_{\mathcal L(\mathbf G_f)},
\end{equation}
where $\lambda_{\mathcal L(\mathbf G_f)}$ is the minimum nonzero $\ell_\infty$-distance of the lattice generated by $\mathbf G_f$.

By the robust MD-CRT~\cite{MD2,mstage-mdcrt}, see also Proposition~2 in Part~I, $\mathbf f^{\mathrm{int}}$ can be robustly recovered from $\tilde{\mathbf r}_1$ and $\tilde{\mathbf r}_2$ if
\begin{equation}\label{eq:tau_f_condition}
\max_{j=1,2}\|\Delta\mathbf r_j\|_\infty < \frac{\Lambda_f}{4}.
\end{equation}
Let $\hat{\mathbf f}^{\mathrm{int}}$ denote the resulting estimate. Then
\begin{equation}\label{eq:fint_error_bound_special_local}
\|\hat{\mathbf f}^{\mathrm{int}}-\mathbf f^{\mathrm{int}}\|_\infty
\le
\max_{j=1,2}\|\Delta\mathbf r_j\|_\infty.
\end{equation}
We then estimate $\mathbf g$ by
\begin{equation}\label{eq:g_hat_noisy_special}
\hat{\mathbf g}
\triangleq
\beta\,\hat{\mathbf f}^{\mathrm{int}}.
\end{equation}
Since $\mathbf g=\beta\mathbf f$ and $\mathbf f^{\mathrm{int}}=\mathbf f+\mathbf e$, we have
\[
\hat{\mathbf g}-\mathbf g
=
\beta(\hat{\mathbf f}^{\mathrm{int}}-\mathbf f^{\mathrm{int}})
+
\beta\mathbf e.
\]
Using \eqref{eq:e_bound_special} and \eqref{eq:fint_error_bound_special_local}, we obtain
\begin{equation}\label{eq:g_error_bound_noisy_special}
\|\hat{\mathbf g}-\mathbf g\|_\infty
\le
|\beta|\max_{j=1,2}\|\Delta\mathbf r_j\|_\infty
+
\frac{|\beta|}{2}.
\end{equation}

The above discussion leads to the following noisy recovery result.

\begin{thm}[Common-Scaling Noisy Recovery]\label{thm:common_beta_noisy}
Assume that the common-scaling condition in \eqref{eq:common_beta} holds, and let $\mathbf R_f$ be an lcrm of $\mathbf N_1$ and $\mathbf N_2$. 
If
$\mathbf f^{\mathrm{int}}\in\mathcal N(\mathbf R_f)$
and \eqref{eq:tau_f_condition} holds,
then $\mathbf f^{\mathrm{int}}$ can be robustly recovered from the noisy vector remainders $\tilde{\mathbf r}_1$ and $\tilde{\mathbf r}_2$. 
Moreover, the estimate
$\hat{\mathbf g}=\beta\,\hat{\mathbf f}^{\mathrm{int}}$
satisfies the reconstruction error bound in \eqref{eq:g_error_bound_noisy_special}.
\end{thm}

Compared with the results in \eqref{eq:s2_noisy_condition}-\eqref{eq:s2_noisy_bound} for the general two-subarray framework, the common-scaling design above remains simpler in the noisy case. 
For the same unambiguous range, we directly recover $\mathbf f^{\mathrm{int}}$ without introducing the additional scaling step \eqref{eq:s2_Abj_def}--\eqref{eq:s2_x_def} used in Section~\ref{s2} to cast $\mathbf g$ into integer-valued form. 
As a result, the sufficient condition for robust recovery no longer contains the first term in \eqref{eq:s2_noisy_condition}, which arises from rounding the scaled parameter vector $\mathbf g$, and the second term, which arises from the quantization of the frequency vector $\mathbf f$. 
Therefore, under the same unambiguous range, the common-scaling design admits a weaker robust recovery condition and thus leads to better robust recovery performance.
The reconstruction error bound in \eqref{eq:g_error_bound_noisy_special} is also tighter, since it does not contain the extra $\frac{1}{k}$ term in \eqref{eq:s2_noisy_bound}. 
Therefore, under the same unambiguous range, the common-scaling design provides a tighter reconstruction error bound.

In summary, the common-scaling design preserves the same unambiguous range as the general two-subarray framework, while leading to weaker robust recovery conditions and tighter reconstruction error bounds in both noiseless and noisy settings. 
As a result, it achieves better robustness under the same unambiguous range.

\section{Performance Comparisons}\label{s4}

In this section, we compare the recovery performance of the proposed planar array framework in the noiseless setting under the same physical constraints. 
We first compare it with a conventional separated scheme under the same platform size and minimum antenna spacing constraints. 
We then compare four representative planar array designs under the same type of constraints to examine how recovery performance depends on both the number of antennas and the array geometry. 
These theoretical comparisons also provide the basis for interpreting the noisy numerical results in Section~\ref{s5}.

\subsection{Comparison with Separated Processing}\label{s4.1}

In this subsection, we compare the proposed planar array framework with a conventional separated scheme based on the existing one-dimensional treatments of the two parameters. 
In the separated scheme, the cross-range shift is recovered by the CRT-based method using two non-uniform linear subarrays in~\cite{gangli}, while the target height is estimated from cross-track interferometric phase as in~\cite{xiaoweili}, which can be thought of as a two-point DFT in the vertical direction. 
The comparison is restricted to the noiseless setting, where no additive noise is considered, as what was considered in the compared one in \cite{gangli}.
To ensure a meaningful comparison, the proposed planar array framework and the separated scheme are both under the same platform size and minimum antenna spacing constraints. 

For the cross-range shift, this allows a direct comparison with the one-dimensional two-subarray design in~\cite{gangli}, including the unambiguous range, the sufficient recovery condition under quantization, and the corresponding reconstruction error bound. 
For the height parameter, however, a complete alignment is not possible, since the separated scheme uses only one additional cross-track receiver rather than a planar two-subarray structure. 
Accordingly, for the height parameter, we compare only the unambiguous range under the same minimum antenna spacing constraint in this subsection.

Since the proposed planar array framework and the separated scheme have the same unambiguous range under the imposed platform size and minimum antenna spacing constraints, the purpose of this comparison is not to enlarge the range itself. 
Instead, it is to highlight a structural advantage of the planar array construction, i.e., its ability to realize a common-scaling design.

To see it clearly, we first recall the relevant one-dimensional formulation in~\cite{gangli}. 
Consider two one-dimensional subarrays, each consisting of $M$ antenna elements, with inter-element spacings $d_1$ and $d_2$, respectively. 
After some standard radar signal processing, the signal in the $i$-th subarray can be written as
\begin{equation}\label{eq:1d_signal_common_scaling}
S_i(m)
=
\exp\!\left(j2\pi \frac{F_i}{M}m\right)\delta(n-n_0-\Delta_{\rm shift}),
\end{equation}
for $m=0,1,\dots,M-1$, where
\begin{equation}\label{eq:1d_Fi_def}
F_i
\triangleq
\frac{M\Delta_x d_i}{R_0\lambda}\,\Delta_{\rm shift},
\qquad i=1,2.
\end{equation}
Hence, in the one-dimensional setting, the DFT modulus is the scalar $M$, and the corresponding scaling factor relating $\Delta_{\rm shift}$ and $F_i$ is
\begin{equation}\label{eq:1d_beta_def}
\beta_i^{(1\mathrm D)}
\triangleq
\frac{\Delta_{\rm shift}}{F_i}
=
\frac{R_0\lambda}{M\Delta_x d_i},
\qquad i=1,2.
\end{equation}
The following proposition shows that, when the two one-dimensional subarrays have the same number of antenna elements, the common-scaling design is impossible.

\begin{prop}\label{prop:1d_no_common_scaling}
Consider the one-dimensional non-uniform linear array scheme in~\cite{gangli}, and assume that the two subarrays have the same number $M$ of antenna elements. 
Then the common-scaling design, $\beta_1^{(1\mathrm D)}=\beta_2^{(1\mathrm D)}$,
is possible only if $d_1=d_2$.
Hence, such a construction is not possible in the considered one-dimensional two-subarray setting.
\end{prop}

\begin{proof}
By \eqref{eq:1d_beta_def}, if
\[
\beta_1^{(1\mathrm D)}=\beta_2^{(1\mathrm D)},
\]
then
\[
\frac{R_0\lambda}{M\Delta_x d_1}
=
\frac{R_0\lambda}{M\Delta_x d_2},
\]
which implies
\[
d_1=d_2.
\]
Therefore, two one-dimensional subarrays with the same number of antenna elements can share the same scaling factor only when their inter-element spacings are equal. 
In that case, the two subarrays are identical, so the supposed two-subarray design degenerates into a single-subarray one. 
Hence, the common-scaling design is unavailable in the one-dimensional setting considered here.
\end{proof}

We now explain why the planar array setting considered in this paper is different. 
For the $j$-th planar subarray, the scaling factor is
\[
\beta_j=\frac{R\lambda}{\alpha_j\det(\mathbf M_j)},
\]
and the number of antenna elements is
\[
|\det(\mathbf N_j)|
=
\alpha_j^2|\det(\mathbf M_j)|.
\]
Suppose that two planar subarrays are required to have the same number of antenna elements and also satisfy the common-scaling condition $\beta_1=\beta_2$. 
Then
\[
\alpha_1^2|\det(\mathbf M_1)|
=
\alpha_2^2|\det(\mathbf M_2)|
\]
and
\[
\alpha_1\det(\mathbf M_1)=\alpha_2\det(\mathbf M_2).
\]
These two relations imply $\alpha_1=\alpha_2$, and hence
\[
\det(\mathbf M_1)=\det(\mathbf M_2).
\]
Therefore, unlike the one-dimensional case, two planar subarrays with the same number of antenna elements may still be constructed from two distinct generator matrices with the same determinant. 
These two subarrays then share the same scaling factor, while their associated matrix moduli $\mathbf{N}_j$ defined in \eqref{eq:s2_Nj_def},
\[
\mathbf N_j=\alpha_j\operatorname{adj}(\mathbf M_j^\top),
\qquad j=1,2,
\]
can still be different. 

To illustrate this additional freedom, consider the integer matrix case. 
For a fixed determinant, one would like to choose two co-prime integer matrices in order to obtain a larger lcrm in the absolute determinant value sense. 
Various constructions of pairwise co-prime integer matrices with the same absolute determinant value have been studied in~\cite{guo25,primematrix} for non-prime and prime integer matrices, respectively. 
Interestingly, for $2\times 2$ co-prime integer matrices with the same absolute determinant value, two matrices are already sufficient to attain the maximal lcrm in the absolute determinant value sense~\cite{mstage-mdcrt}. 
Since a larger lcrm implies a larger unambiguous range, adding a third such matrix does not further increase the size of the unambiguous range.  
This provides an additional justification for focusing on the two-subarray case in the common-scaling design considered in this Part II. 
On the other hand, the general $J$-subarray framework in Part~I remains essential for more general multi-subarray designs, and the present two-subarray common-scaling design is only one particular design within the general framework.
How to construct an optimal planar array under practical physical constraints within the general framework remains an open problem and will be investigated in our future work.

Therefore, the planar array framework can realize a common-scaling pair, and this structural advantage removes the extra condition required in the one-dimensional scheme for robust recovery under quantization.

To make the above structural difference concrete, we next compare a representative common-scaling planar array design with a representative separated one-dimensional scheme. 
The two examples are chosen to satisfy the same platform size and minimum antenna spacing constraints, so that they have the same unambiguous range. 
The comparison then focuses on whether the proposed planar array framework yields a weaker recovery condition and a tighter reconstruction error bound, while the corresponding noisy case performance will be evaluated in Section~\ref{s5}.

\begin{figure*}[htbp]
    \centering
    \begin{subfigure}[t]{0.37\linewidth}
        \centering
        \includegraphics[width=\linewidth]{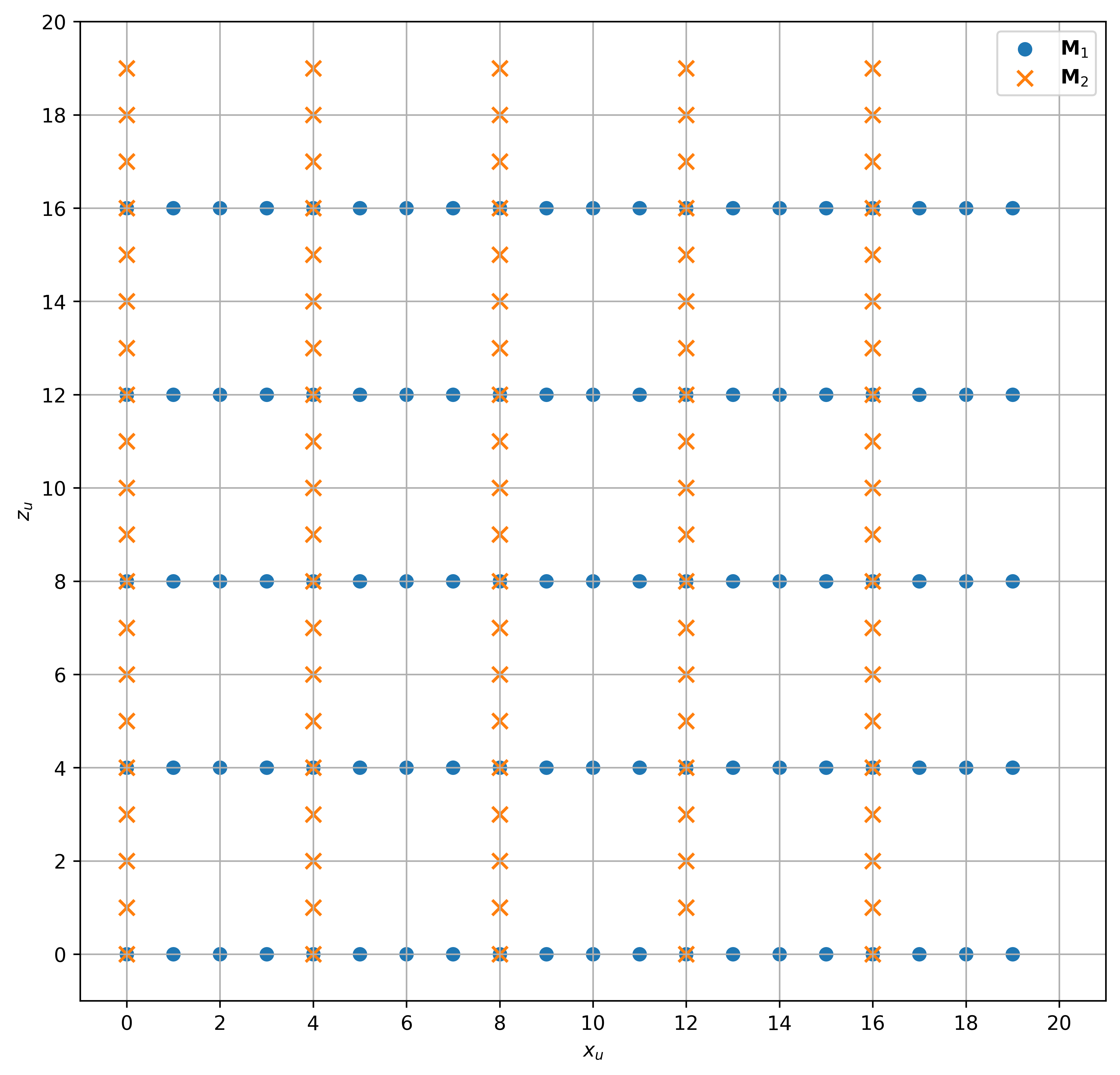}
        \caption{Common-scaling planar array design.}
        \label{fig:common_scaling_array}
    \end{subfigure}
    \begin{subfigure}[t]{0.37\linewidth}
        \centering
        \includegraphics[width=\linewidth]{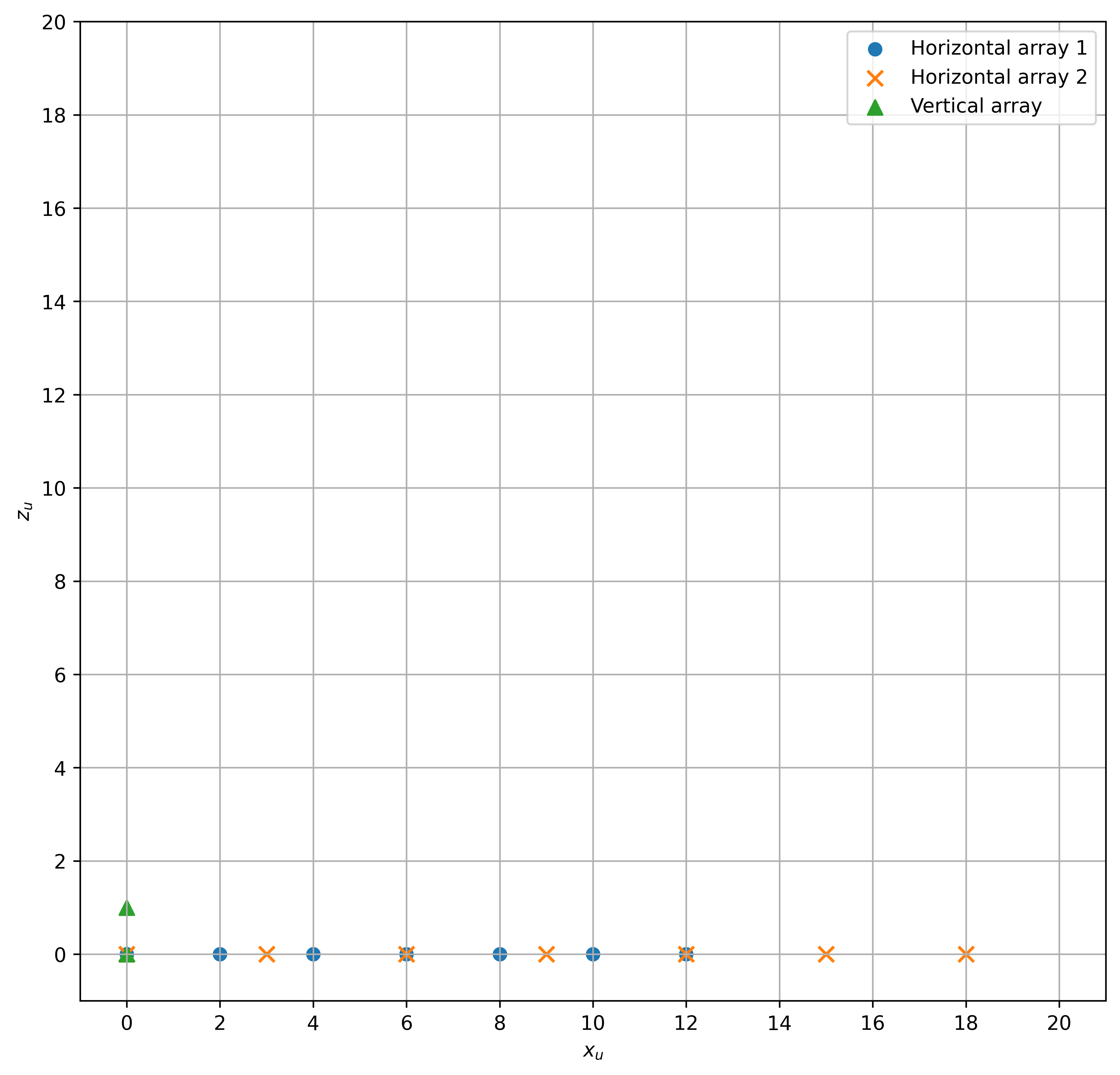}
        \caption{Separated one-dimensional design.}
        \label{fig:separated_1d_array}
    \end{subfigure}
    \caption{Representative array designs in Examples~\ref{ex:common_scaling_planar} and \ref{ex:separated_1d_baseline}.}
    \label{fig:2D_vs_sep1D_arrays}
\end{figure*}

\begin{example}\label{ex:common_scaling_planar}
We first consider the common-scaling planar array shown in Fig.~\ref{fig:2D_vs_sep1D_arrays}(a), whose elements lie in a $19\times 19$ platform in the $x$--$z$ plane with minimum antenna spacing $1$.
Two subarrays are constructed by choosing
\[
\mathbf M_1=
\begin{pmatrix}
1 & 0\\
0 & 4
\end{pmatrix},
\qquad
\mathbf M_2=
\begin{pmatrix}
4 & 0\\
0 & 1
\end{pmatrix},
\]
with $\alpha_1=\alpha_2=5$. Since $\det(\mathbf M_1)=\det(\mathbf M_2)=4$, the two subarrays have the same scaling factor
\begin{equation}\label{eq:planar_beta_example}
\beta_1=\beta_2=\beta=\frac{R\lambda}{20}.
\end{equation}
By \eqref{eq:s2_Nj_def}, the corresponding integer matrix moduli are
\[
\mathbf N_1=
\begin{pmatrix}
20 & 0\\
0 & 5
\end{pmatrix},
\qquad
\mathbf N_2=
\begin{pmatrix}
5 & 0\\
0 & 20
\end{pmatrix}.
\]

Let
\begin{equation}\label{eq:planar_Rf_example}
\mathbf R_f=\operatorname{lcrm}(\mathbf N_1,\mathbf N_2)=
\begin{pmatrix}
-20 & 0\\
0 & 20
\end{pmatrix}.
\end{equation}
Hence, the unambiguous range for $\mathbf f^{\mathrm{int}}$ is
\[
\mathcal N(\mathbf R_f)
=
\left\{
\begin{bmatrix}
-k_1\\
k_2
\end{bmatrix}
:\;
k_1,k_2\in\{0,1,\dots,19\}
\right\}.
\]
Since $\mathbf g=\beta\mathbf f$, the corresponding unambiguous range for $\mathbf g$ is
\begin{equation}\label{eq:planar_range_example}
\Delta_{\rm shift}\Delta_x\in[0,R\lambda),
\qquad
z_0\in[0,R\lambda).
\end{equation}

By Theorem~\ref{thm:common_scaling_noiseless}, in the noiseless case, the reconstruction error satisfies
\begin{equation}\label{eq:planar_g_error_example}
\|\hat{\mathbf g}-\mathbf g\|_\infty\le \frac{\beta}{2}=\frac{R\lambda}{40}.
\end{equation}
Equivalently,
\begin{equation}\label{eq:planar_param_error_example}
|\hat{\Delta}_{\rm shift}-\Delta_{\rm shift}|
\le
\frac{R\lambda}{40\Delta_x},
\qquad
|\hat z_0-z_0|
\le
\frac{R\lambda}{40}.
\end{equation}
\end{example}

\begin{example}\label{ex:separated_1d_baseline}
We next consider the separated one-dimensional design shown in Fig.~\ref{fig:2D_vs_sep1D_arrays}(b). Its elements also lie in a $19\times 19$ platform in the $x$--$z$ plane, and the minimum antenna spacing is $1$.

For the cross-range shift, we use the horizontal two-subarray construction in \cite{gangli}. The two one-dimensional subarrays both  have $7$ elements and spacings $d_1=2$ and $d_2=3$. The corresponding scalar moduli are
\begin{equation}\label{eq:sep_moduli_example}
L_1=\frac{R\lambda}{\Delta_x d_1}=\frac{R\lambda}{2\Delta_x},
\qquad
L_2=\frac{R\lambda}{\Delta_x d_2}=\frac{R\lambda}{3\Delta_x}.
\end{equation}
Hence,
\begin{equation}\label{eq:sep_shift_range_example}
\Delta_{\rm shift}\Delta_x\in\left[0,\Delta_x\operatorname{LCM}(L_1,L_2)\right)
=
[0,R\lambda).
\end{equation}
According to \cite{gangli}, robust recovery of $\Delta_{\rm shift}$ under DFT frequency grid quantization further requires
\begin{equation}\label{eq:sep_condition_example}
N_{\rm DFT}>L_1+L_2=\frac{5R\lambda}{6\Delta_x},
\end{equation}
where $N_{\rm DFT}$ denotes the DFT size used for each linear subarray after zero-padding, and without zero-padding, it reduces to the number of antennas, i.e., $7$, in each one-dimensional subarray. In this case, the reconstruction error satisfies
\begin{equation}\label{eq:sep_shift_error_example}
|\hat{\Delta}_{\rm shift}-\Delta_{\rm shift}|
\le
\frac{5R\lambda}{24N_{\rm DFT}\Delta_x}.
\end{equation}

For the height direction, we adopt the single cross-track receiver model in \cite{xiaoweili}. The vertical spacing is $1$. Hence,
\begin{equation}\label{eq:sep_height_range_example}
z_0\in[0,R\lambda).
\end{equation}
\end{example}

We now compare Examples~\ref{ex:common_scaling_planar} and \ref{ex:separated_1d_baseline}. First, the two schemes have the same unambiguous range. Indeed, by \eqref{eq:planar_range_example}, the proposed planar array design yields
\[
\Delta_{\rm shift}\Delta_x\in[0,R\lambda),
\qquad
z_0\in[0,R\lambda),
\]
and by \eqref{eq:sep_shift_range_example} and \eqref{eq:sep_height_range_example}, the separated one-dimensional scheme yields exactly the same range. Under the present platform size and the minimum antenna spacing constraints, the comparison is therefore not about enlarging the unambiguous range.

The essential difference appears in the robust recovery requirement for the cross-range shift under quantization. 
In the separated one-dimensional scheme, robust recovery requires the additional condition \eqref{eq:sep_condition_example}. 
In contrast, for the proposed planar array design, the common-scaling pair is constructed directly within the planar array model, so Theorem~\ref{thm:common_scaling_noiseless} applies without introducing any extra condition of this kind. 
Therefore, under the same unambiguous range, the proposed planar array design admits a weaker robust recovery requirement, so robust recovery is easier to achieve.

A similar difference appears in the reconstruction error bounds. For the proposed planar array design, \eqref{eq:planar_param_error_example} gives bounds for both cross-range shift and height errors. For the separated one-dimensional scheme, \eqref{eq:sep_shift_error_example} gives a bound only for the shift error. Moreover, in this example, the shift reconstruction error bound of the planar array design is smaller when zero padding is excluded from the separated one-dimensional array design, i.e., $N_{\rm DFT}=7$, because zero padding is a post-processing operation rather than a structural advantage of the one-dimensional construction and can also be applied to the two-dimensional DFT in the planar scheme, as discussed in Part I.

In summary, the two schemes have the identical unambiguous range under the same constraints, but they differ in recovery structure. 
The proposed planar array framework realizes the common-scaling design, so it does not require the additional condition needed in the separated one-dimensional scheme for robust recovery. 
As a result, under the same unambiguous range, robust recovery is easier to achieve, and the shift reconstruction error bound is smaller in this example. 
These features also help explain the better robustness observed for the proposed framework in the noisy case simulations of Section~\ref{s5}.

\subsection{Comparison of Planar Array Designs}\label{s4.2}

In this subsection, we compare four representative array designs under the same physical constraints: a non-separable planar array design formed by two subarrays, a uniform planar array design, a separable planar array design formed by two subarrays, and a pair of one-dimensional arrays with one along the horizontal direction and the other along the vertical direction.

All the four array designs can be viewed as subarrays selected from the same dense planar array that occupies a $95\times95$ planar region with unit spacing. Hence, the dense array contains $96^2=9216$ elements. Each of the four designs satisfies the same minimum antenna spacing requirement of $5$.
They are compared here as four different array designs under the same physical constraints. 
The purpose of this comparison is to show that recovery performance depends not only on the number of antennas but also on the geometry of the selected antenna set. 
Here, the term ``separable'' refers to the array geometry, namely, that the corresponding generator matrix is diagonal, and should not be confused with the ``separated'' processing scheme in Section~\ref{s4.1}. 
As in the previous subsection, we restrict the analysis to the noiseless setting and focus on the unambiguous range and the corresponding reconstruction error bound. 
The noisy case will be examined later in Section~\ref{s5}.

The four representative array designs are shown in Fig.~\ref{fig:array_comparison}, and their details are given below.

\begin{figure*}[htbp]
    \centering

    \begin{subfigure}[t]{0.24\textwidth}
        \centering
        \includegraphics[width=\linewidth]{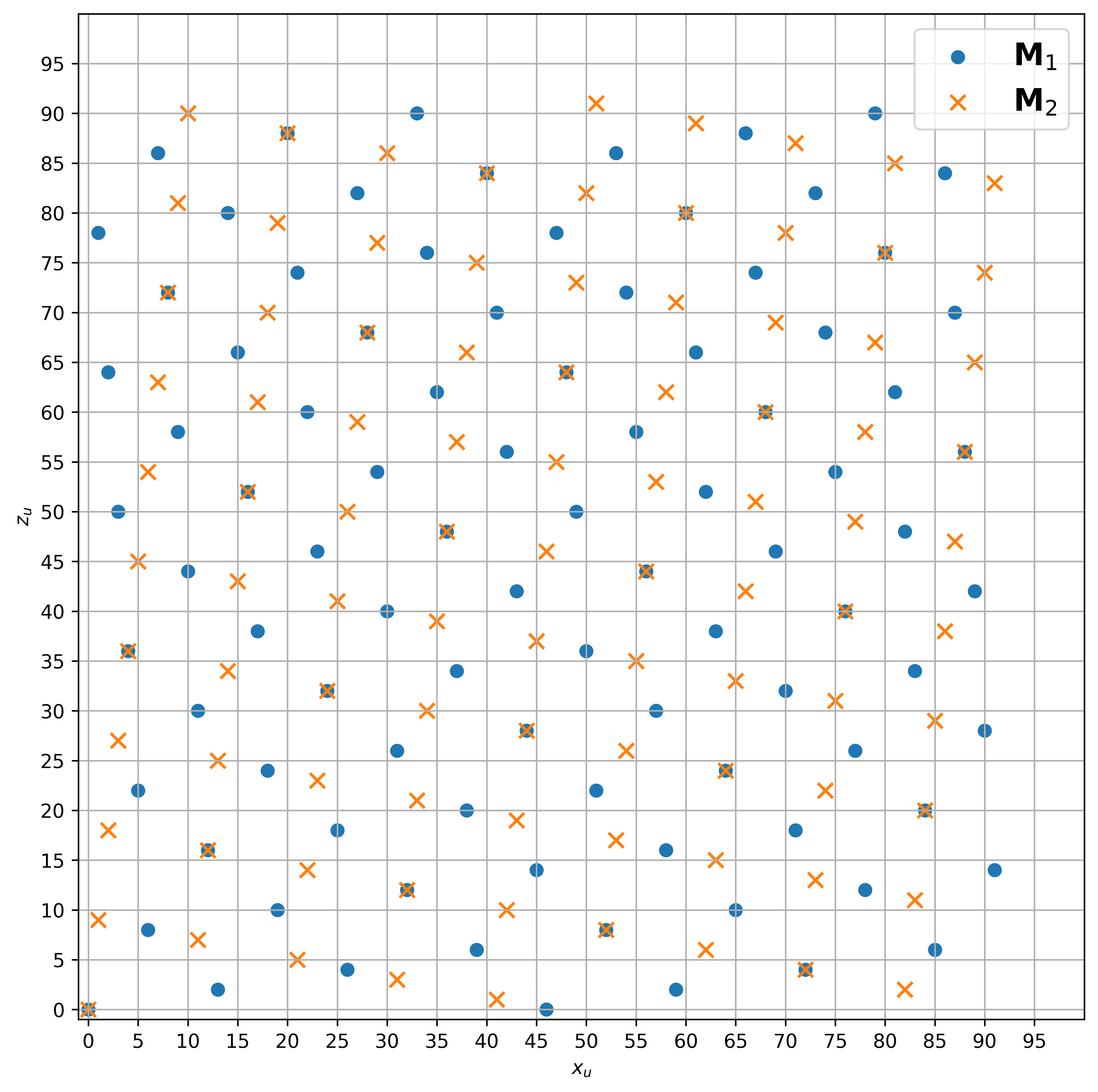}
        \caption{Non-separable planar array design (161 elements).}
        \label{fig:array_a}
    \end{subfigure}
    \begin{subfigure}[t]{0.24\textwidth}
        \centering
        \includegraphics[width=\linewidth]{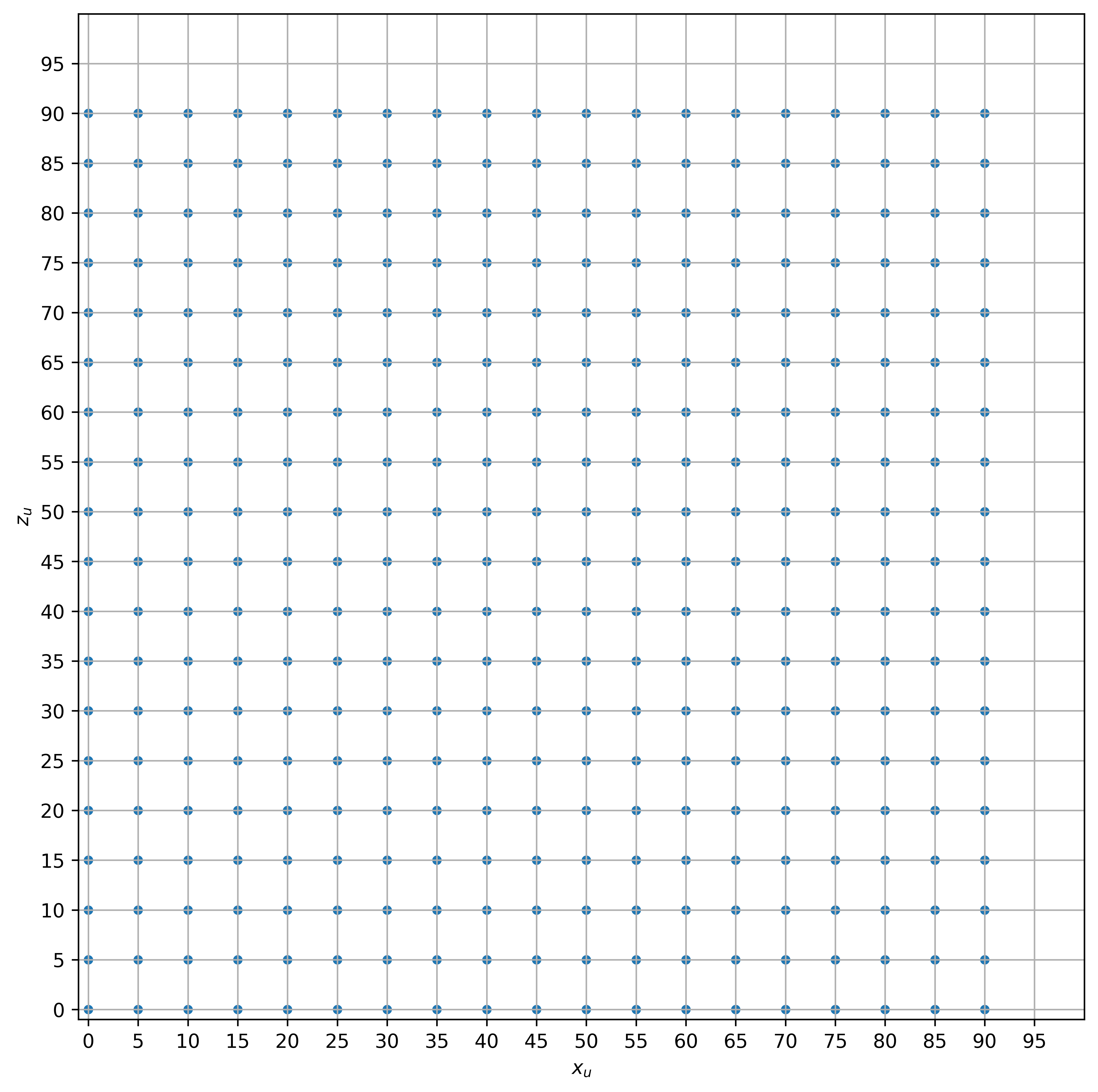}
        \caption{Uniform planar array design (361 elements).}
        \label{fig:array_b}
    \end{subfigure}
    \begin{subfigure}[t]{0.24\textwidth}
        \centering
        \includegraphics[width=\linewidth]{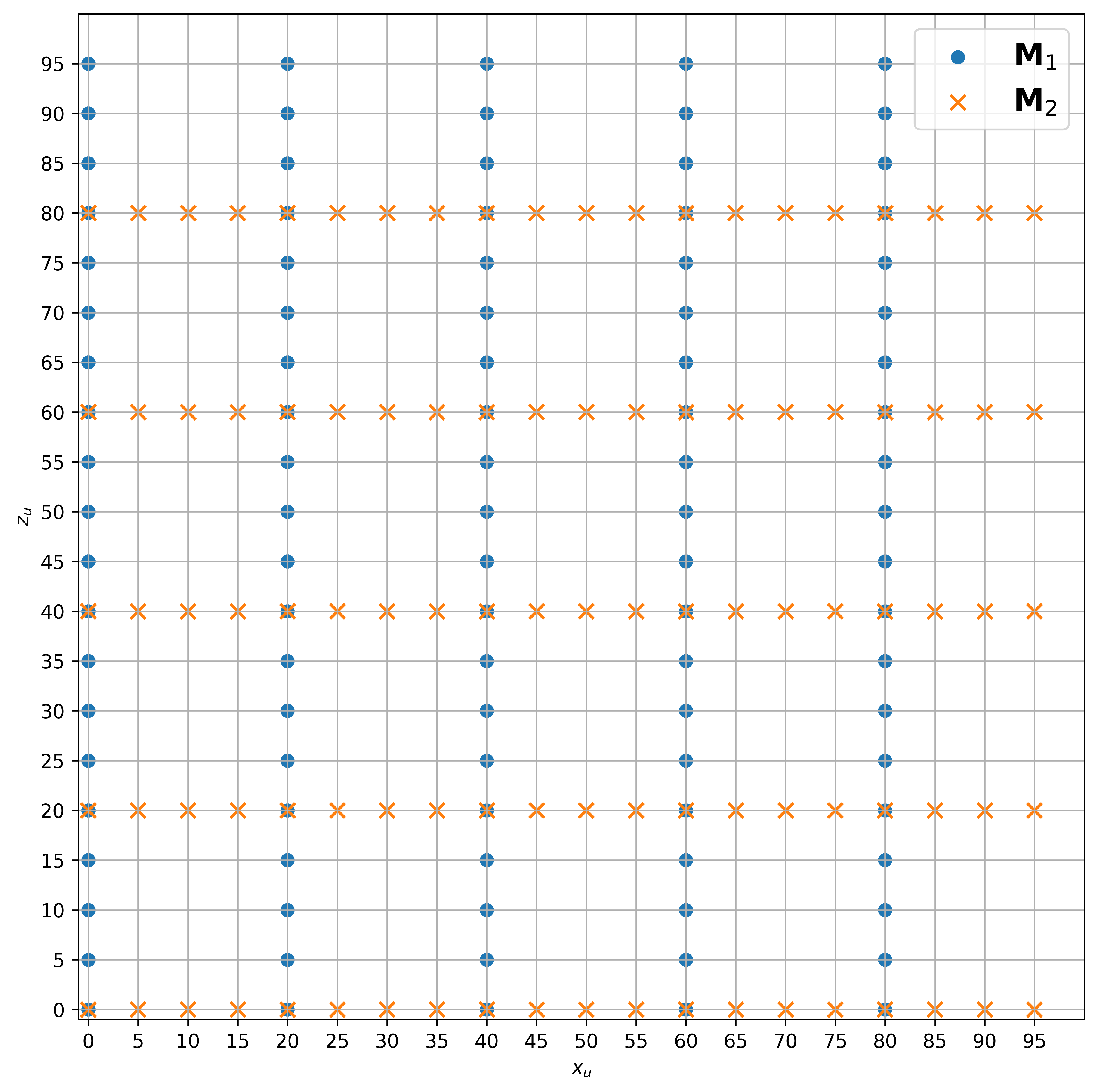}
        \caption{Separable planar array design (175 elements).}
        \label{fig:array_c}
    \end{subfigure}
    \begin{subfigure}[t]{0.24\textwidth}
        \centering
        \includegraphics[width=\linewidth]{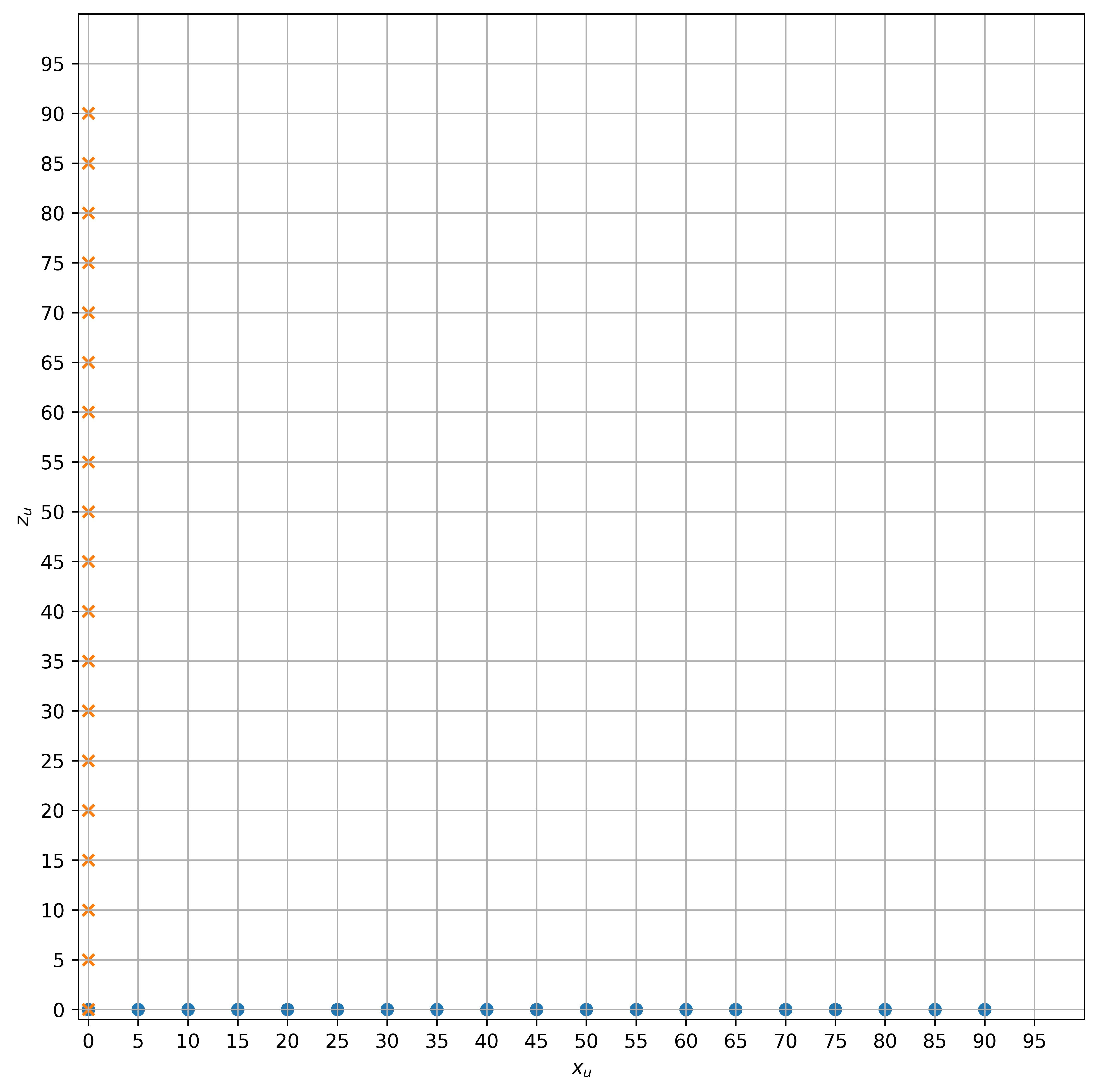}
        \caption{Pair of one-dimensional arrays (37 elements).}
        \label{fig:array_d}
    \end{subfigure}

    \caption{Four array designs under the same $95\times95$ planar platform and the same minimum antenna spacing of $5$.}
    \label{fig:array_comparison}
\end{figure*}

\begin{example}[Non-separable planar array design]\label{ex:nonseparable_planar}
We first consider the non-separable planar array design shown in Fig.~\ref{fig:array_comparison}(a). 
Its elements lie in a $95\times95$ platform and occupy a $91\times91$ region. 
The minimum antenna spacing is $5$. The array is constructed by choosing
\[
\mathbf M_1=
\begin{pmatrix}
46 & -33\\
0 & 2
\end{pmatrix},
\qquad
\mathbf M_2=
\begin{pmatrix}
51 & -10\\
-1 & 2
\end{pmatrix},
\]
with $\alpha_1=\alpha_2=1$. Since $\det(\mathbf M_1)=\det(\mathbf M_2)=92$, the two subarrays have the same scaling factor
\begin{equation}\label{eq:nsep_beta_example}
\beta_1=\beta_2=\beta=\frac{R\lambda}{92}.
\end{equation}
By \eqref{eq:s2_Nj_def}, the corresponding integer matrix moduli are, respectively,
\[
\mathbf N_1=
\begin{pmatrix}
2 & 0\\
33 & 46
\end{pmatrix},
\qquad
\mathbf N_2=
\begin{pmatrix}
2 & 1\\
10 & 51
\end{pmatrix}.
\]
The resulting array contains $161$ elements.

Let
\begin{equation}\label{eq:nsep_R_example}
\mathbf R=\operatorname{lcrm}(\mathbf N_1,\mathbf N_2)
=
4
\begin{pmatrix}
1 & 0\\
5 & 23
\end{pmatrix}
\mathbf U,
\end{equation}
where $\mathbf U$ is unimodular. Hence, according to \eqref{eq:lcrm_common_scaling}, the unambiguous range of $\mathbf g$ is
\begin{equation}\label{eq:nsep_range_example}
\mathcal N(\beta\mathbf R)
=
\mathcal N\!\left(
\frac{R\lambda}{23}
\begin{pmatrix}
1 & 0\\
5 & 23
\end{pmatrix}
\mathbf U
\right).
\end{equation}

By Theorem~\ref{thm:common_scaling_noiseless}, in the noiseless case,
\begin{equation}\label{eq:nsep_error_example}
\|\hat{\mathbf g}-\mathbf g\|_\infty
\le
\frac{\beta}{2}
=
\frac{R\lambda}{184}.
\end{equation}
\end{example}

\begin{example}[Uniform planar array design]\label{ex:dense_planar}
We next consider the uniform planar array design shown in Fig.~\ref{fig:array_comparison}(b). 
It lies in a $95\times95$ platform, occupies a $90\times90$ region, and has minimum antenna spacing $5$. 
This example is chosen for comparison with the previous non-separable design, which occupies a $91\times91$ planar region. 
Under a comparable platform size and the same minimum antenna spacing constraint, it gives the densest square-lattice arrangement.
The array is constructed by choosing
\[
\mathbf M=5\mathbf{I},
\]
with $\alpha=\frac{19}{5}$. By \eqref{eq:s2_Nj_def}, the corresponding matrix modulus is
\[
\mathbf N=19\mathbf I.
\]
The resulting array contains $361$ elements, much more than the previous two non-separable subarray design.

Since only one modulus is involved, the unambiguous range of $\mathbf g$ is directly given by
\begin{equation}\label{eq:dense_range_example}
\mathcal N(\beta \mathbf N)
=
\mathcal N\!\left(
\frac{R\lambda}{5}\mathbf I
\right),
\qquad
\beta=\frac{R\lambda}{\alpha\det(\mathbf M)}=\frac{R\lambda}{95}.
\end{equation}

In the noiseless case,
\begin{equation}\label{eq:dense_error_example}
\|\hat{\mathbf g}-\mathbf g\|_\infty
\le
\frac{\beta}{2}
=
\frac{R\lambda}{190}.
\end{equation}
\end{example}

\begin{example}[Separable planar array design]\label{ex:separable_planar}
We next consider the separable planar array design shown in Fig.~\ref{fig:array_comparison}(c). It occupies a $95\times95$ platform and has minimum antenna spacing $5$. Compared with the non-separable design, its occupied region is slightly larger in order to satisfy both the common-scaling setting and the minimum antenna spacing constraint. The array is constructed by choosing two diagonal integer matrices
\[
\mathbf M_1 =
\begin{pmatrix}
5 & 0\\
0 & 20
\end{pmatrix},
\qquad
\mathbf M_2 =
\begin{pmatrix}
20 & 0\\
0 & 5
\end{pmatrix},
\]
with $\alpha_1=\alpha_2=1$. Since $\det(\mathbf M_1)=\det(\mathbf M_2)=100$, the two subarrays have the same scaling factor
\begin{equation}\label{eq:sep2d_beta_example}
\beta_1=\beta_2=\beta=\frac{R\lambda}{100}.
\end{equation}
By \eqref{eq:s2_Nj_def}, the corresponding matrix moduli are
\[
\mathbf N_1 =
\begin{pmatrix}
20 & 0\\
0 & 5
\end{pmatrix},
\qquad
\mathbf N_2 =
\begin{pmatrix}
5 & 0\\
0 & 20
\end{pmatrix}.
\]
The resulting array contains $175$ elements.

Let
\begin{equation}\label{eq:sep2d_R_example}
\mathbf R=\operatorname{lcrm}(\mathbf N_1,\mathbf N_2)=20\mathbf I.
\end{equation}
Hence, according to \eqref{eq:lcrm_common_scaling}, the unambiguous range of $\mathbf g$ is
\begin{equation}\label{eq:sep2d_range_example}
\mathcal N(\beta\mathbf R)=\mathcal N\!\left(\frac{R\lambda}{5}\mathbf I\right).
\end{equation}

By Theorem~\ref{thm:common_scaling_noiseless}, in the noiseless case,
\begin{equation}\label{eq:sep2d_error_example}
\|\hat{\mathbf g}-\mathbf g\|_\infty \le \frac{\beta}{2}=\frac{R\lambda}{200}.
\end{equation}
\end{example}

\begin{example}[Pair of one-dimensional arrays]\label{ex:two_1d_arrays}
We finally consider the pair of one-dimensional arrays shown in Fig.~\ref{fig:array_comparison}(d). 
One array lies on the $x$-axis and the other lies on the $z$-axis. 
Each array has spacing $5$ and contains $19$ elements, so the total number of elements is $37$. 
The design lies in a $95\times95$ platform and occupies a $90\times90$ region.

This design recovers the two components of $\mathbf g$ separately by two one-dimensional DFT-based constructions, one along the $x$ direction and the other along the $z$ direction. 
Hence, the unambiguous range of $\mathbf g$ is
\begin{equation}\label{eq:two1d_range_example}
\mathcal N\!\left(
\frac{R\lambda}{5}\mathbf I
\right).
\end{equation}

In the noiseless case, the reconstruction error satisfies
\begin{equation}\label{eq:two1d_error_example}
\|\hat{\mathbf g}-\mathbf g\|_\infty
\le
\frac{R\lambda}{190}.
\end{equation}
\end{example}

We now discuss the above four array designs in terms of the unambiguous range, the reconstruction error bound, and the role of antenna number and array geometry in robustness. 

\subsubsection{Unambiguous Range Comparison}

First consider the unambiguous range. For the uniform planar array design, the separable planar array design, and the pair of one-dimensional arrays, the recovery regions in \eqref{eq:dense_range_example}, \eqref{eq:sep2d_range_example}, and \eqref{eq:two1d_range_example} are identical. This is because these three designs are all separable, so the unambiguous range in each dimension is determined by the minimum antenna spacing along that dimension. Since the minimum spacing is $5$ in both dimensions for all three designs, their unambiguous ranges are all given by
\[
\mathcal N\!\left(\frac{R\lambda}{5}\mathbf I\right).
\]

The non-separable planar array design is different from the above three designs. 
Since it is non-separable, its unambiguous range is not determined per dimension by the corresponding minimum antenna spacings, as in the other three cases. 
Instead, it is determined by the matrix $\mathbf R$ through the region $\mathcal N(\beta\mathbf R)$ in \eqref{eq:nsep_range_example}. 
As a result, its unambiguous range is not rectangular.

To compare the size of this unambiguous range with those of the other three designs, it is more natural to compare the numbers of integer vectors contained in the corresponding unambiguous ranges. 
This is because of the following reason. 
In the integer-valued vector case, an FPD as a set of integer vectors is an unambiguous range as we explained earlier. 
For the general non-integer valued vector case, a vector is approximated by an integer vector by rounding, and then the integer vector will be in the consideration. 
Therefore, the number of integer vectors in an FPD provides a natural discrete measure of the size of the unambiguous range. 
This avoids a direct geometric comparison between 2D regions of different shapes.
For simplicity, assume that $R\lambda/5$ and $R\lambda/23$ are both integers. 
Then, for the three separable designs, the number of integer vectors in the unambiguous range is
$\left(\frac{R\lambda}{5}\right)^2=\frac{(R\lambda)^2}{25}$.
For the non-separable planar array design, by \eqref{eq:nsep_beta_example} and \eqref{eq:nsep_R_example}, the number of integer vectors in the unambiguous range is
\[
|\det(\beta\mathbf R)|=\frac{(R\lambda)^2}{23},
\]
which is larger than $(R\lambda)^2/25$. 
Therefore, the non-separable planar array design provides a larger unambiguous range than the other three designs. 
Although the improvement from $1/25$ to $1/23$ may not seem large, this is only an illustrative example and is not meant to be optimal. Better and optimal designs may exist.

This difference comes from the array geometry. 
For separable designs with minimum antenna spacing $5$, the smallest diagonal generator matrix in the absolute determinant value sense is $5\mathbf I$. 
In contrast, a non-diagonal generator matrix may still satisfy the same minimum antenna spacing requirement while having determinant $23$. 
As a result, under the common-scaling requirement, a non-separable geometry can provide a larger unambiguous range while using fewer antennas and occupying less platform area than the separable design in Example~\ref{ex:separable_planar}.

\subsubsection{Comparison of Reconstruction Error Bounds}

We next compare the reconstruction error bounds. 
We first consider the three separable designs, namely, the uniform planar array design, the separable planar array design, and the pair of one-dimensional arrays. 
For separable designs, the bound in each dimension is determined by the spacing and the number of antennas along that dimension. 
In Examples~\ref{ex:dense_planar} and \ref{ex:two_1d_arrays}, the spacing is $5$ and the number of antennas in each dimension is $19$. 
In Example~\ref{ex:separable_planar}, the spacing remains $5$, but the number of antennas in each dimension becomes $20$. 
This is why Example~\ref{ex:separable_planar} yields a smaller reconstruction error bound than Examples~\ref{ex:dense_planar} and \ref{ex:two_1d_arrays}. 
In contrast, the uniform planar array design in Example~\ref{ex:dense_planar} and the pair of one-dimensional arrays in Example~\ref{ex:two_1d_arrays} have the same reconstruction error bound, because their spacing and antenna counts are identical in each dimension.
Therefore, simply replicating the same one-dimensional structure along the other dimension does not improve the bound in either dimension.

We next compare the non-separable and separable common-scaling planar array designs in Examples~\ref{ex:nonseparable_planar} and \ref{ex:separable_planar}. 
Under the common-scaling design, the noiseless reconstruction error bound is determined by the common scaling factor
\[
\beta=\frac{R\lambda}{\alpha\det(\mathbf M)}.
\]
Therefore, a larger $\alpha$ or a larger absolute determinant value leads to a smaller bound. 
This explains why the separable planar array design in Example~\ref{ex:separable_planar} yields a smaller reconstruction error bound than the non-separable planar array design in Example~\ref{ex:nonseparable_planar}.

It is also important to note that the above bound is only a worst-case bound. 
More precisely, the reconstruction error bound of $\mathbf g$ is $\beta$ multiplied by the quantization error of $\mathbf f$, while the bound here is obtained by replacing the quantization error with its worst-case value. 
However, since $\mathbf f$ is obtained from the same $\mathbf g$ after scaling by $1/\beta$, a smaller $\beta$ does not necessarily imply a smaller actual quantization error of $\mathbf f$. 
In some cases, when $\beta$ becomes smaller, the quantization error of $\mathbf f$ may become larger. 
Therefore, when the quantization error term dominates, the actual reconstruction errors could be ordered in the different way from the worst-case bound. 
This effect will be seen later in the simulations.

\subsubsection{Antenna Number and Robustness}

We finally comment on the numbers of antenna elements. 
Comparing Examples~\ref{ex:dense_planar} and \ref{ex:two_1d_arrays}, we see that the uniform planar array design uses many more antennas, while the unambiguous range and the noiseless reconstruction error bound are the same. 
This shows that the noiseless metrics considered in this subsection do not fully reflect the benefit of using more antennas. 
As will be seen in Section~\ref{s5}, when additive noise is present, more antennas can improve robustness.

At the same time, robustness does not depend only on the number of antennas. 
Our simulations will also show that the non-separable design in Example~\ref{ex:nonseparable_planar} is more robust than the separable design in Example~\ref{ex:separable_planar}. 
This suggests that, in addition to the number of antennas, the array geometry also affects robustness. 
In robust MD-CRT, a non-separable matrix may provide better robustness than a separable one, and this is related to the geometry of the lattice generated by the gcld and we refer to \cite{mstage-mdcrt} for more details. 
The above discussion is based on only one representative example, and the construction of an optimal array design remains open.

\section{Simulation Results}\label{s5}

In this section, we present numerical results for the noisy case. 
The simulations are intended to illustrate how the proposed planar antenna array framework performs under additive noise. 
To keep the comparison focused on the recovery behavior itself, we adopt idealized settings for DFT peak detection and range-related information, rather than performing a full end-to-end simulation from raw transmitted and received radar signals. 
The experiments are organized to follow the two comparisons in Section~\ref{s4}. 
The first subsection compares the proposed planar array framework with the separated scheme, and the second compares different array designs in the noisy case.

Unless otherwise stated, the radar altitude is set to $H=4000~\text{m}$, the platform velocity is set to $v=200~\text{m/s}$, the cross-range resolution is set to $\Delta_x=2~\text{m}$, and the wavelength is set to $\lambda=0.05~\text{m}$. 
The terrain, road, and target position are generated numerically. 
Specifically, a road is first placed on the terrain surface, and a point target is then selected on the road according to a prescribed height and a normalized distance along the road. 
The target velocity is chosen along the local road direction with a prescribed speed, so that both target location and motion are determined by the generated road geometry.

\subsection{Comparison with the Separated Scheme}

In this subsection, we compare the proposed planar array framework with the separated scheme in the noisy case. 
The simulation setup follows the comparison in Section~\ref{s4.1}. 
The target height is set to about $95~\text{m}$. 
The target speed is set to $13~\text{m/s}$, and the target is placed at a point whose normalized position along the road is $0.3$. 
This gives the target position $\mathbf r_0=[x_0,y_0,z_0]^{\top}=[175.409811,\ 263.516684,\ 94.868794]^{\top}~\text{m}$ and the target velocity $\mathbf v_0=[v_x,v_y,v_z]^{\top}=[3.246923,\ 9.450347,\ 8.315553]^{\top}~\text{m/s}$. 
The SNR varies from $-10~\text{dB}$ to $10~\text{dB}$ with a step size of $1~\text{dB}$, and $2000$ Monte Carlo trials are used for each SNR.

In all the simulations in this section, additive white Gaussian noise is introduced at the level of the compensated complex measurements, and the SNR is defined as the ratio of the mean signal power to the noise power for the complex signals in, for example, \eqref{eq:s2_signal_model}. 
Since each planar subarray in the proposed method contains $100$ antennas, the separated one-dimensional scheme is also implemented with a $100$-point zero-padded DFT for a fair comparison. 
To show the behavior of the two parameters more clearly, we plot the RMSEs of $\Delta_{\rm shift}$ and $z_0$ separately.

\begin{figure*}[htbp]
    \centering
    \begin{subfigure}[t]{0.45\textwidth}
        \centering
        \includegraphics[width=\linewidth]{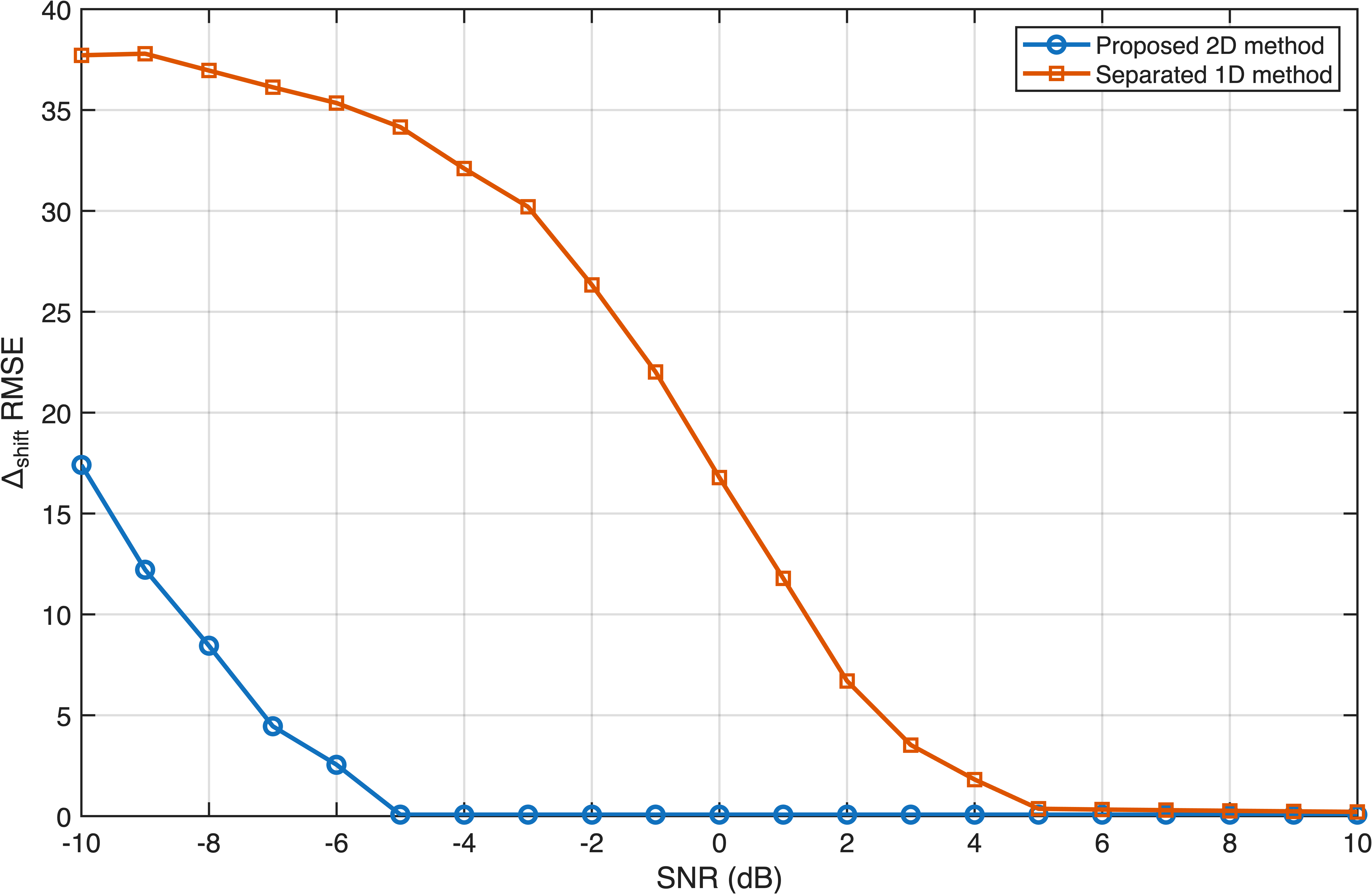}
        \caption{RMSE of $\Delta_{\rm shift}$ versus SNR.}
        \label{fig:sim1_shift_rmse}
    \end{subfigure}
    \hfill
    \begin{subfigure}[t]{0.45\textwidth}
        \centering
        \includegraphics[width=\linewidth]{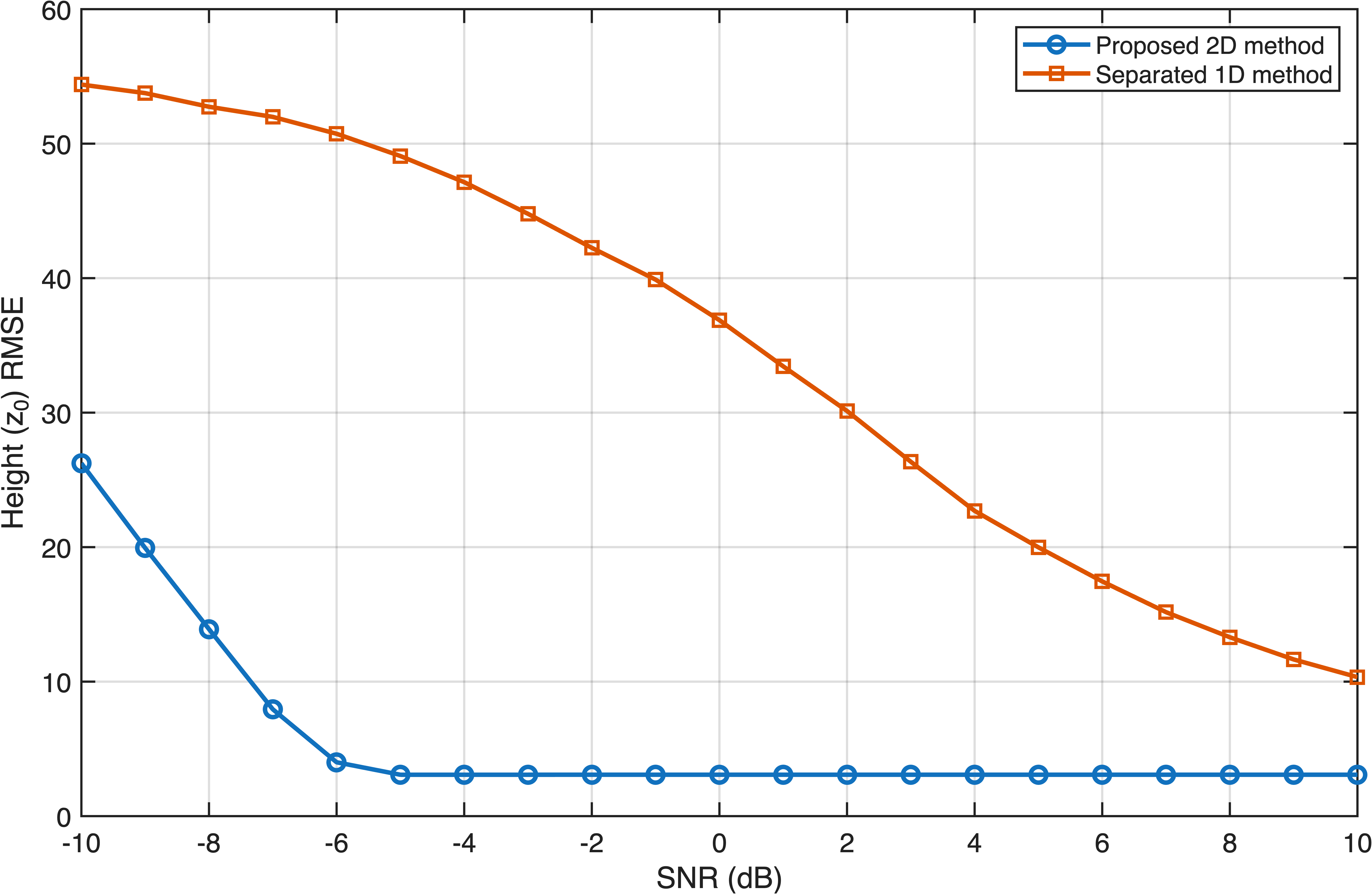}
        \caption{RMSE of $z_0$ versus SNR.}
        \label{fig:sim1_z_rmse}
    \end{subfigure}
    \caption{Comparison between the proposed 2D method and the separated 1D method.}
    \label{fig:sim1_rmse}
\end{figure*}

The results are shown in Fig.~\ref{fig:sim1_rmse}. 
For both $\Delta_{\rm shift}$ and $z_0$, the proposed method achieves a much smaller RMSE than the separated scheme over the whole SNR range. 
The difference is especially clear in the low- and moderate-SNR regions, where the proposed method reaches a small error level much earlier, while the separated scheme improves more slowly. 
This behavior is consistent with the analysis in Section~\ref{s4.1}, where the two schemes were shown to have the same unambiguous range in this example, but the proposed common-scaling planar array design admits a weaker robust recovery condition and a smaller theoretical error bound for the shift. 
These advantages imply that, under noise, robust recovery is easier to achieve and the resulting recovery error is smaller, which is consistent with the lower RMSE observed in Fig.~\ref{fig:sim1_rmse}. 
The noisy case results therefore show that, under the same horizontal platform length constraint, the proposed framework remains much more robust, although the number of antennas may also affect the noisy case performance.

The gain for $z_0$ is also clear. 
In the separated scheme, the height is estimated directly from the cross-track interferometric phase.
This can be thought of as a two-point DFT in the vertical direction and is more sensitive to noise.  
In contrast, the proposed method yields a much smaller RMSE for $z_0$ over the whole SNR range. 
This indicates that even when two schemes have the same theoretical unambiguous range, their practical recovery performance can still be significantly different.

\subsection{Comparison Among Different Array Designs}

In this subsection, we compare the four array designs in Section~\ref{s4.2} in the noisy case. 
Since the four designs do not have exactly the same unambiguous range, we choose a target that lies in the common unambiguous range of all of them. 
In this way, every design is tested within its own unambiguous range. 
The simulation setup follows the comparison in Section~\ref{s4.2}. 
The target height is set to about $35~\text{m}$. 
The target speed is set to $4~\text{m/s}$, and the target is placed at a point whose normalized position along the road is $0.3$. 
This gives the target position $\mathbf r_0=[x_0,y_0,z_0]^{\top}=[143.228351,\ 169.850773,\ 34.869242]^{\top}~\text{m}$ and the target velocity $\mathbf v_0=[v_x,v_y,v_z]^{\top}=[1.135671,\ 3.305432,\ 1.945346]^{\top}~\text{m/s}$. 
The SNR varies from $-10~\text{dB}$ to $5~\text{dB}$ with a step size of $1~\text{dB}$, and $2000$ Monte Carlo trials are used for each SNR. 
Since the purpose here is to compare the overall recovery performance of the four designs, we use the mean of $\|\hat{\mathbf g}-\mathbf g\|_\infty$ as the error measure. 
This quantity reflects the reconstruction error of the worse component in each trial.

\begin{figure}[htbp]
    \centering
    \includegraphics[width=0.9\linewidth]{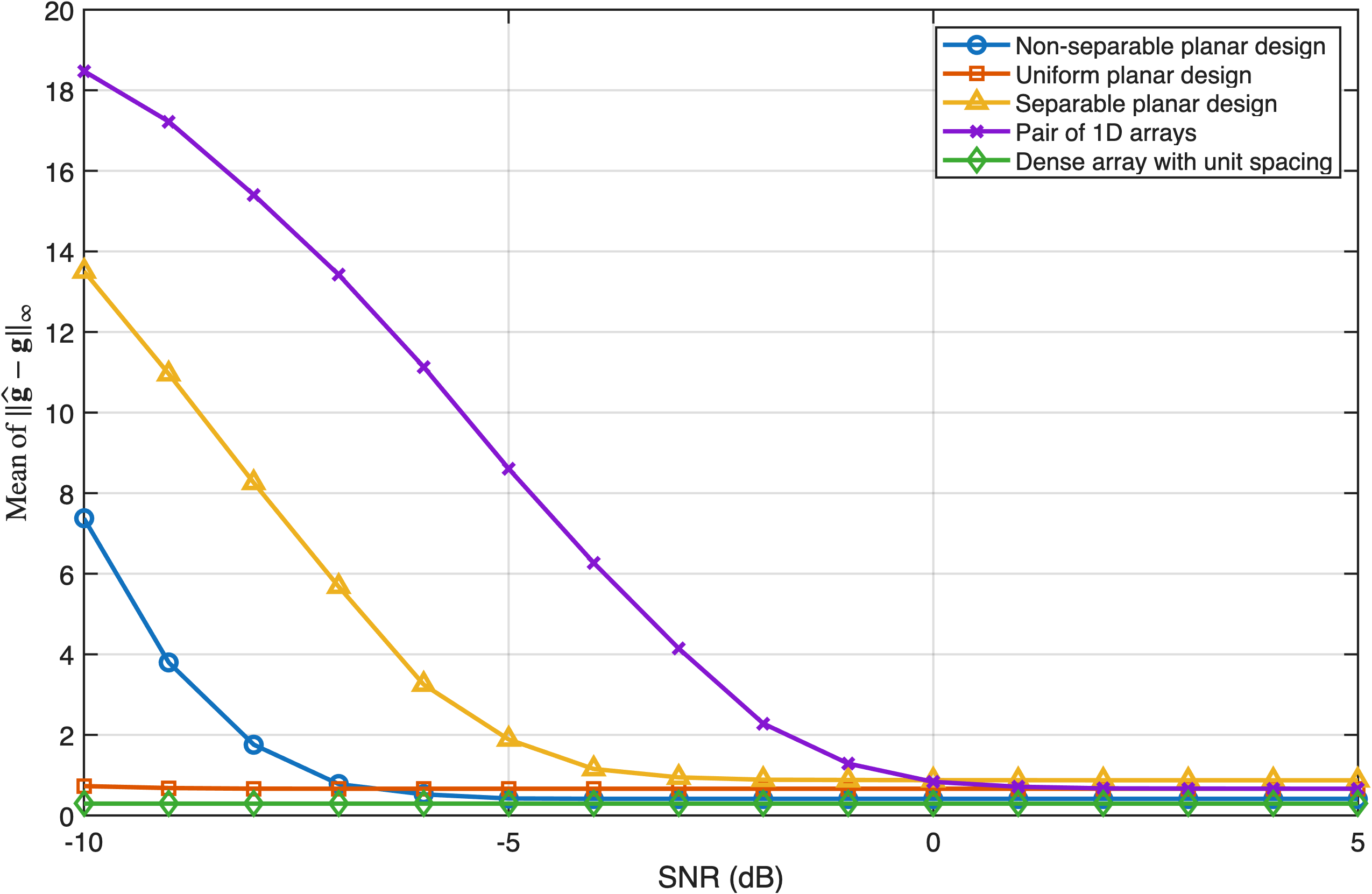}
    \caption{Mean of $\|\hat{\mathbf g}-\mathbf g\|_\infty$ versus SNR for four array designs.}
    \label{fig:sim2_g_linf_mean}
\end{figure}

The results are shown in Fig.~\ref{fig:sim2_g_linf_mean}.  
The uniform planar array design gives the smallest error at very low SNR and remains stable over the whole SNR range, with a nearly flat error floor. 
The non-separable planar array design starts from a higher error level, but its error decreases much more rapidly as the SNR increases, and it soon becomes the best-performing design among the four. 
By contrast, the separable planar array design and the pair of one-dimensional arrays have much larger errors at low SNR, and their errors decrease more slowly. 

As the SNR enters the high-SNR region, the gap among the four designs becomes much smaller. 
In this region, the non-separable planar array design gives the smallest error, the uniform planar array design and the pair of one-dimensional arrays are nearly the same, and the separable planar array design gives the largest error. 
This ordering is different from the noiseless reconstruction error bound  comparison in Section~\ref{s4.2}. 
As discussed at the end of that subsection, the theoretical bound is only a worst-case bound. 
In the noisy case setting, different values of $\beta$ lead to different quantization errors in the corresponding integer vector $\mathbf f$. 
When this effect dominates, the actual reconstruction errors could be ordered differently from the worst-case bound.

As noted earlier, all the four designs can be viewed as subarrays selected from the same $95\times95$ dense planar array with unit spacing, and each of the four satisfies the same minimum spacing requirement of $5$. 
For reference, we also plot the performance of this dense planar array. 
It gives the smallest error over the whole SNR range, which is consistent with the fact that the four sparse designs are all selected from it. 
Although the dense planar array performs the best, it uses the largest number of antennas, namely $9216$, which is much larger than that of any of the above four compared arrays.

The above results of the four planar array designs show that noisy case performance is determined by both antenna number and array geometry. 
To understand this behavior more clearly, we consider these two effects separately. 
The role of antenna number is especially clear from the comparison between the uniform planar array design and the pair of one-dimensional arrays. 
Although these two designs have the same unambiguous range and the same reconstruction error bound in the noiseless case analysis, the uniform planar array design is much more robust under noise. 
The non-separable and separable planar array designs lie between these two extremes in terms of antenna number, and their noisy case performance also lies between them. 
This shows that a larger number of antennas can indeed improve robustness.

\begin{figure*}[htbp]
    \centering
    \begin{subfigure}[t]{0.45\textwidth}
        \centering
        \includegraphics[width=\linewidth]{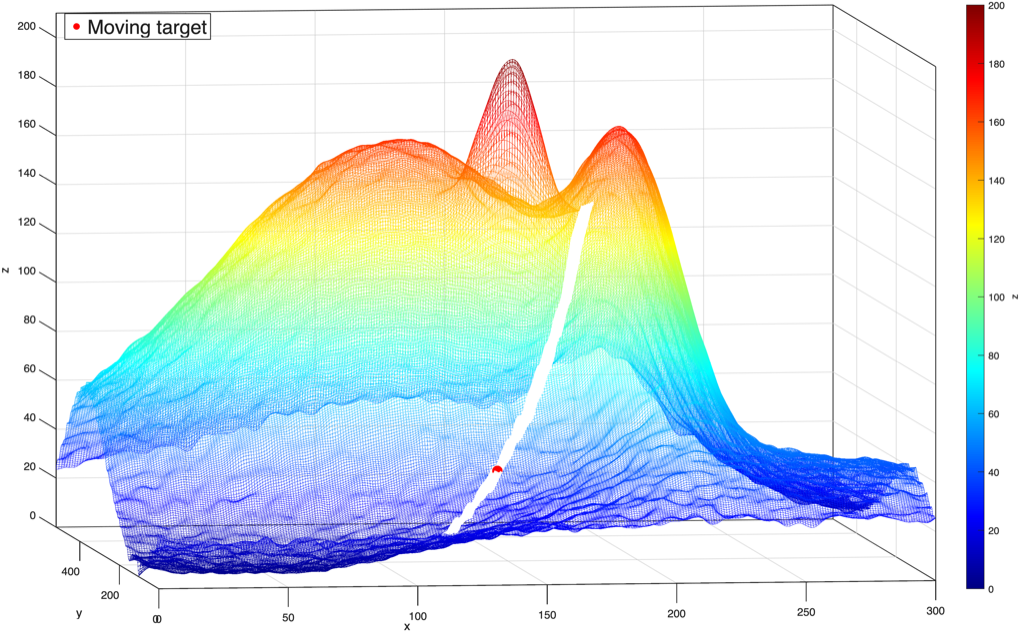}
        \caption{Non-separable planar array design.}
        \label{fig:sim2_3d_1}
    \end{subfigure}
    \hfill
    \begin{subfigure}[t]{0.45\textwidth}
        \centering
        \includegraphics[width=\linewidth]{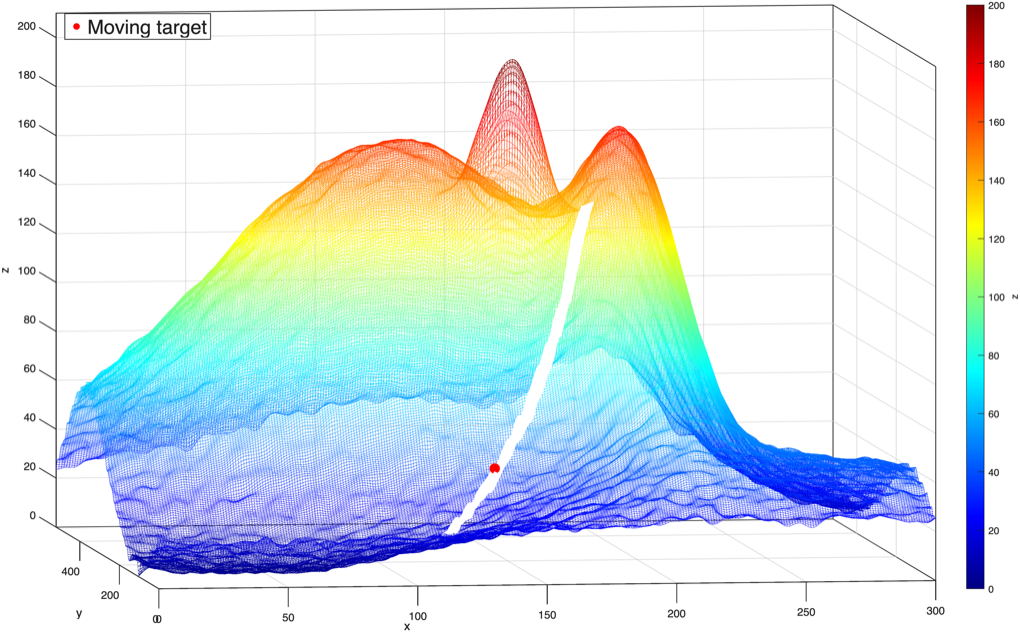}
        \caption{Uniform planar array design.}
        \label{fig:sim2_3d_2}
    \end{subfigure}

    \vspace{0.5em}

    \begin{subfigure}[t]{0.45\textwidth}
        \centering
        \includegraphics[width=\linewidth]{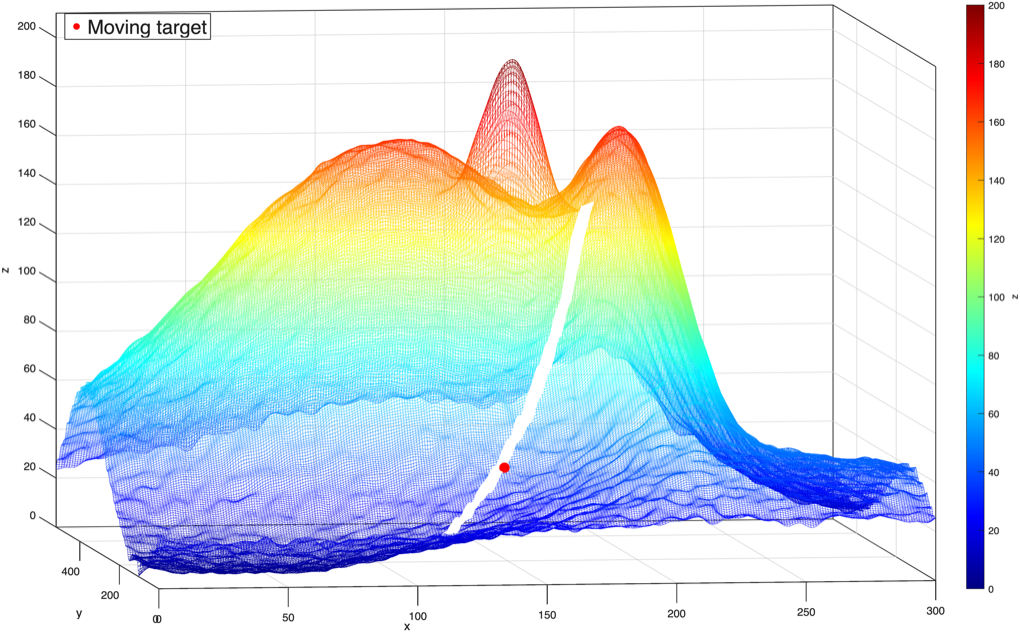}
        \caption{Separable planar array design.}
        \label{fig:sim2_3d_3}
    \end{subfigure}
    \hfill
    \begin{subfigure}[t]{0.45\textwidth}
        \centering
        \includegraphics[width=\linewidth]{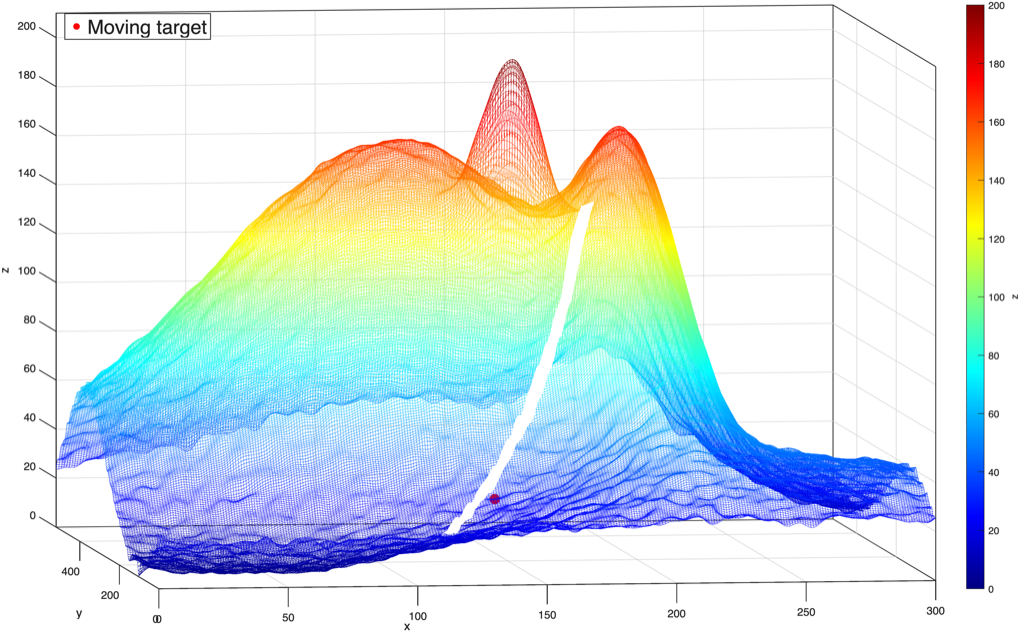}
        \caption{Pair of one-dimensional arrays.}
        \label{fig:sim2_3d_4}
    \end{subfigure}
    \caption{Recovered target locations at $\text{SNR}=-5~\text{dB}$ for four array designs. The reconstructed targets are marked by red points.}
    \label{fig:sim2_3d}
\end{figure*}

However, antenna number is not the only factor. 
The comparison between the non-separable planar array design and the separable planar array design is especially important here. 
In the noisy case, both designs rely on robust MD-CRT for recovery. 
Although Theorem~\ref{thm:common_beta_noisy} gives a sufficient condition for robust recovery, it may not fully characterize the actual robustness behavior. 
A more accurate understanding would require the necessary and sufficient condition in~\cite{MD2}, where robustness is linked to the lattice geometry induced by the gcld matrix. 
How to use that condition to precisely characterize the robustness difference between different planar array designs remains open. 
In the above, both designs satisfy the same minimum antenna spacing requirement of $5$. 
Moreover, as discussed earlier, the separable design corresponds to the diagonal gcld matrix with determinant $25$, whereas the non-separable design considered here corresponds to a non-diagonal gcld matrix with determinant $23$. 
This suggests that the observed robustness difference is related to the array geometry rather than antenna number alone. 
Further discussion on the robustness difference between non-separable and separable designs may be found in~\cite{mstage-mdcrt}. 
The above discussion is based on only one representative example, and the problem of optimal array design for robustness remains open as well.

To further illustrate the recovery behavior, we also show the reconstructed target locations in three-dimensional space at $\mathrm{SNR}=-5~\mathrm{dB}$. These plots provide a visual counterpart to the quantitative comparison in Fig.~\ref{fig:sim2_g_linf_mean}. In this experiment, the range coordinate $y_0$ is assumed to be known. The reconstructed target positions are marked by red points.

The results are shown in Fig.~\ref{fig:sim2_3d}. 
For the non-separable planar array design and the uniform planar array design, the recovered target locations remain on the road. 
For the separable planar array design, the recovered target is slightly shifted away from the road. 
In contrast, for the pair of one-dimensional arrays, the recovered target is shifted much more noticeably and appears on the lower-right side of the road. 
This is consistent with their performance at $\mathrm{SNR}=-5~\mathrm{dB}$ in Fig.~\ref{fig:sim2_g_linf_mean}.

Overall, the numerical results in this section support the theoretical analysis in Section~\ref{s4}. 
They show that the proposed planar array framework yields better robustness than the separated scheme, and that, among different array designs, both antenna number and array geometry play essential roles in noisy case  recovery.

\section{Conclusion}\label{s6}

In Part~I of this two-part work, we established a general multidimensional Chinese remainder theorem (MD-CRT) framework for moving target SAR imaging with planar antenna arrays, including the associated matrix-modulus formulation, ambiguity resolution framework, and robust recovery analysis. 
In this Part II, we focused on performance analysis and planar array design within the framework.

We studied a common-scaling two-subarray design, under which the two planar subarrays share the same scaling factor in the corresponding MD-CRT, which is not possible for one-dimensional linear antenna arrays. 
This design preserves the unambiguous range of the general formulation while simplifying the recovery procedure in both noiseless and noisy settings. 
As a result, the associated robust recovery conditions become weaker and the corresponding reconstruction error bounds become tighter, leading to better robustness.

Theoretical analysis and numerical results showed that, under the same unambiguous range, the proposed planar array framework offers clear advantages over the conventional separated processing in terms of robustness. 
They also showed that recovery performance depends not only on the number of antennas but also on the array geometry. 
In particular, non-separable planar array geometries can provide better robustness than separable ones when their antenna numbers are comparable. 
These results suggest that planar array geometry should be regarded as an integral part of ambiguity resolution and robust parameter estimation.
Future work includes the design of optimized planar array geometries under practical platform size, spacing, and noise constraints.

\end{document}